\documentclass[10pt,twocolumn,journal]{IEEEtran}
\usepackage{latexsym}
\usepackage{amsfonts}
\usepackage{amsbsy}
\usepackage{amsmath,amssymb}
\usepackage{times}
\usepackage{graphicx}
\usepackage{enumerate}
\usepackage[usenames]{color}
\usepackage[dvips]{pstcol}
\usepackage{epstopdf}
\usepackage{cite}
\usepackage{amsmath}
\usepackage{amssymb}
\usepackage{amsfonts}
\usepackage{epsfig}
\usepackage{psfrag}
\usepackage{color}
\usepackage{xcolor}
\usepackage{amsfonts, bm}
\usepackage{epstopdf}
\usepackage{cite}
\usepackage{color}
\usepackage{xcolor}
\usepackage{subfig}
\usepackage{algorithm}
\usepackage{stfloats}
\usepackage{amsthm}

\newtheorem{lemma}{\textbf{Lemma}}
\newtheorem{assumption}{\textbf{Assumption}}
\newtheorem{remark}{\textbf{Remark}}
\newtheorem{definition}{\textbf{Definition}}

\newtheorem{theorem}{\textbf{Theorem}}

\input epsf

\begin{document}
\title{Stochastic Transceiver Optimization in Multi-Tags Symbiotic Radio Systems}

\author{\IEEEauthorblockN{Xihan Chen, \IEEEmembership{Student Member,~IEEE}, Hei Victor Cheng, \IEEEmembership{Member,~IEEE}, Kaiming Shen, \IEEEmembership{Student
Member,~IEEE},\\
An Liu, \IEEEmembership{Senior Member,~IEEE}, and Min-Jian Zhao, \IEEEmembership{Member,~IEEE}}

\thanks{
X. Chen, A. Liu and M. Zhao are with the Department of Information Science and Electronic Engineering, Zhejiang University, Hangzhou, China (e-mail: chenxihan@zju.edu.cn;
anliu@zju.edu.cn;  mjzhao@zju.edu.cn).

H.V. Cheng, and
K. Shen is with the Edward S. Rogers Sr. Department of Electrical and Computer Engineering, University of Toronto, Toronto, ON M5S 3G4 (e-mail: hei.cheng@utoronto.ca,
kshen@ece.utoronto.ca).
}
}

\maketitle

\begin{abstract}
Symbiotic radio (SR) is emerging as a spectrum- and energy-efficient communication paradigm for future passive Internet-of-things (IoT), where some single-antenna
backscatter devices, referred to as Tags, are parasitic in an active primary transmission. The primary transceiver is designed to assist both direct-link (DL) and
backscatter-link (BL) communication. In multi-tags SR systems, the transceiver designs become much more complicated due to the presence of  DL and inter-Tag interference,
which further poses new challenges to the availability and reliability of DL and BL transmission. To overcome these challenges, we formulate the stochastic optimization of
transceiver design as the general network utility maximization problem (GUMP). The resultant problem is a stochastic multiple-ratio fractional non-convex problem, and
consequently challenging to solve. By leveraging some fractional programming techniques, we tailor a surrogate function with the specific structure and subsequently develop
a batch stochastic parallel decomposition (BSPD) algorithm, which is shown to converge to stationary solutions of the GNUMP. Simulation results verify the effectiveness of
the proposed algorithm by numerical examples in terms of the achieved system throughput.
\end{abstract}


\IEEEpeerreviewmaketitle

\section{Introduction} \label{Intro}
The Internet-of-Things (IoT) \cite{IoT_survey} is envisioned to be a key application scenario in the  next-generation wireless networks, which is conceived for connecting a
plethora of devices together. Specifically, the IoT greatly enhances quality of our daily life and enables various emerging services such as health-care, smart home, and many
others \cite{IoT_benefit}. Despite these potential merits, there are two major challenges in realizing such huge gain in practical IoT systems \cite{IoT_Challenge}. First, the
battery life of IoT devices is usually limited due to their restricted size, which requires more cost-effective maintenance without frequent battery replacement/rechargement
\cite{QinTao}. On the other hand, the number of IoT devices is expected to grow in large numbers, while the radio spectrum resource is limited and will be insufficient to
cater the needs of all IoT devices \cite{FDD_CBSC}. Therefore, it is highly imperative to design spectral- and energy-efficient radio technologies to support practical IoT
systems.

To alleviate the performance bottleneck, one promising solution is to leverage the ambient backscatter communications (AmBC) \cite{binary_mod}, where the ambient backscatter
devices, referred to as Tags \cite{Tag}, can transmit over the surrounding ambient radio frequency (RF) signals  without requiring active RF components. Moreover, the AmBC
shares the same spectrum with the legacy system, and does not occupy any additional spectrum, which thereby significantly improves the spectral utilization efficiency
\cite{AMBC_survey}. However, due to the nature of spectrum sharing, the direct-link (DL) transmission imposes severe interference on the backscatter-link (BL) transmission
\cite{DLI}, which further limits the overall system performance. To overcome this issue, a growing body of literatures have recently proposed various methods. In
\cite{FS_wifi}, the frequency shifting (FS) technique was first integrated with AmBC to suppress the  DL interference (DLI). This framework was further developed in
\cite{FS_inter} to enable efficient inter-technology backscatter transmission. Nevertheless, the implementations of above FS techniques require the additional spectrum
resources, which is not suitable for IoT deployments due to the spectrum scarcity problem. Bypassing the above problems, \cite{MA_BSC} combines AmBC and
multiple-input-multiple-output (MIMO) to further improve the spectral efficiency of IoT networks as well as effectively mitigate the DLI. Moreover, a cooperative receiver with
multiple antennas was proposed in \cite{Cooperative_BSC} to recover signals from both the DL and BL transmission. In \cite{OFDM}, a novel joint design for backscattering
waveform and receiver detector was devised exploiting the repeating structure of the ambient orthogonal frequency division multiplexing (OFDM) signals due to the use of
cyclic prefix.

Different from treating the legacy signal as an interference in conventional AmBC systems, a novel technique, called symbiotic radio (SR), was first proposed in
\cite{SymRadio} to assist both the DL and the BL transmission. Specifically, the BL transmission shares not only the RF source and the spectrum resource but also the receiver
with DL transmission in SR systems. Hence, the SR designs have the following advantages: On one hand, the BL transmission can help the DL transmission through providing
additional multipath. On the other hand, the DL transmission offers transmission opportunity to Tags \cite{rate_region}.  Moreover, conventional AmBC harvests energy from ambient signals to support the circuit operation, while in the SR system, we focus on the symbiotic relationship between primary and BL transmissions, and thus the circuit operation can be supported by an internal power source.
 Due to these desirable features, SR systems have
recently drawn considerable interests from the research community. In \cite{SymRadio}, two novel joint transceivers was designed to maximize the weighted sum rate and minimize
the transmit power for the single-Tag SR systems, respectively. Furthermore, the SR system with fading channels was considered in \cite{ResourceA_BSC}, where the transmit
power at the DL transmission and the reflection coefficient at the Tag were jointly optimized to maximize the ergodic weighted sum rate of DL and BL transmission under
long-term/short-term transmit power constraint. To further reduce the transmission latency and improve the spectral efficiency, the authors of \cite{FD_BSC} amalgamated
full-duplex techniques with the SR systems and tackled the self-interference by relying on clever interference cancellation. In addition, the authors of \cite{Noma_BSC}
invoked the non-orthogonal multiple access (NOMA) technology  for the SR systems to enable multiple devices share the allotted spectrum in the most effective manner.

However, it is worth emphasizing that the existing transceiver designs for single-Tag SR systems cannot be directly applied to the multi-Tags SR systems. In  multi-Tags SR
systems, the transceiver designs become much more complicated due to the presence of inter-Tag interference, which further poses new challenges to the availability and
reliability of DL and BL transmission. Moreover, the previous works focus on maximizing rate performance for Tags, which may not be a suitable in the SR regime. The reason is that
Tag cannot perform complicated adaptive modulation coding (AMC) but only simple coding and modulation scheme such as OOK with simple channel coding (e.g., Hamming code) with a fixed data rate, due to the limited hardware capability. As such, it is more desirable to guarantee certain signal-to-interference-plus-noise ratio (SINR) for all Tags so that all Tags can achieve an acceptable quality-of-service (QoS)/ bit error rate (BER).
To the best of our knowledge,  \textbf{although the conventional multi-Tags AmBC system has been studied in the literature, the transceiver design of the multi-Tags SR system based on the SINR utility maximization has not yet been addressed in the literature.}

Motivated by the above concerns, we propose a stochastic transceiver design for the downlink of the multi-Tags SR systems, to alleviate the performance bottleneck caused by
the DLI and inter-Tag interference. In such design, the primary transmitter (PT) first employs the transmit beamforming to support both the DL and BL transmission. Moreover,
the primary receiver (PR) applies the receive beamforming to suppress the DLI and separate out the backscattered signal from multiple Tags. Then, we jointly optimize the
transmit beamforming at the PT as well as the receive beamforming at the PR by maximizing a general utility function of the expected SINR of both the primary and backscatter
transmissions, under practical optimization constraints.  In particular, we also focus on a novel SINR utility function as an example, which is designed based on the barrier method \cite{NP_OPT}
such that by maximizing it, the PR's expected SINR can be maximized and meanwhile, each Tag's expected SINR can be guaranteed to exceed certain threshold. Note that the resulting problem formulation is amenable to network optimization because it avoids complicated stochastic nonconvex QoS constraints, but meanwhile guaranteeing the favorable SINR of all Tags with high accuracy.
The designed multi-Tags SR system has a great potential for applications in future green IoT systems such as wearable sensor networks and smart homes. For example, a smartphone simultaneously recovers information from both a WiFi access point and domestic sensors for smart-home applications.
 However, these are also several technical challenges in the implementation of this architecture:
\begin{itemize}
\item \textbf{Stochastic Multiple-ratio Fractional  Non-convex Optimization}: The joint optimization of the transmit beamforming and the receive beamforming belongs to
    stochastic non-convex optimization, which is difficult to solve. Specifically, the objective function is nonconvex and contains the expectation operators, which cannot
    be computed analytically and accurately. In addition, the resulting optimization problem is also typically a function of multiple-ratio fractional programming (FP),
    where the optimization variables appear in both the numerator and the denominator, leading to an NP-hard problem.

\item \textbf{Lack of Real-time Tags' Symbol Information}: In practice, it is difficult to obtain the real-time Tags' symbol information (TSI) due to the limited coherence
    time and the huge signaling overhead. Therefore, it is more reasonable to consider a design based on TSI statistics, with the reduced feedback signaling.

\item \textbf{Convergence Analysis}: It is very important to establish the convergence of the algorithm. However, this is non-trivial for a stochastic multiple-ratio
    fractional nonconvex optimization problem, due to the following reasons. First, the objective function of the formulated optimization problem is nonconvex. Second, it is difficult  to estimate accurately the corresponding expected value.
\end{itemize}

To address the above challenges, we propose a \emph{batch stochastic parallel decomposition} (BSPD) algorithm to solve such stochastic multiple-ratio fractional  non-convex
problem.
Different from the traditional sample average approximation (SAA) method \cite{SAA_lec}, it does not require to collect a large number of samples of the random system states before solving the resulting problem and process all samples per iteration, which alleviates the performance bottleneck caused by huge memory requirement and computational complexity. In particular,
we tailor a surrogate function by exploiting the special fractional structure of the formulated optimization problem, which replaces the expected values with properly chosen incremental sample estimates of it and linearizes the nonconvex part. Based on the well-designed surrogate function and properly chosen step-size sequence, such incremental estimates are expected to be more and more accurate as the number of iterations increases.
Consequently, we establish convergence of the proposed BSPD algorithm to stationary solutions in a mini-batch fashion. Moreover, it should be emphasized that the convergence and complexity of the proposed BSPD algorithm heavily depend on the choice of surrogate function. The well-designed surrogate function will help to speed up the convergence speed by preserving the structure of the original problem and provide some other potential advantages as elaborated in Remark \ref{advSurrogate}.
Finally, simulation results verify the
advantages of the proposed algorithms over the baselines.

The rest of paper is organized as follows. In Section \ref{sys_mod}, we give the system model for the downlink of a multi-Tags SR system and
formulate the joint optimization of stochastic transceiver as a stochastic non-convex optimization problem.
The proposed BSPD algorithm and the associated convergence proof are presented in Section \ref{EGUO}. The simulation results are given in Section \ref{sec_simulation} to
verify the advantages of the proposed solution, and the conclusion is given in Section \ref{sec_con}.

\emph{Notations}: Scalar, vectors, and matrices are respectively denoted by lower case, boldface lower case, and boldface upper case letters. $\mathbf{I}_m$ represents an
identity matrix with dimension $m\times m$. For a matrix $\mathbf{A}$, $\mathbf{A}^T$, $\mathbf{A}^{\ast}$, $\mathbf{A}^H$, $\mathbf{A}^{-1}$, and
$\mathrm{Tr}(\mathbf{A})$ denote its transpose, conjugate, conjugate transpose, inverse, cholesky decomposition, and trace, respectively. For a vector $\bm{a}$, $\|\bm{a}\|$
represents its Euclidean norm. Further, $\mathbb{E}[\cdot]$ denotes expectation operation, $\triangleq$ denotes definition, and the operator $\mathrm{vec}(\cdot)$ stacks the
elements of a matrix in one long column vector. $\mathbb{C}^{m\times n}$ ($\mathbb{R}^{m\times n}$) denotes the space of $m\times n$ complex (real) matrices.
$\mathcal{CN}(\delta,\sigma^2)$ represents a complex Gaussian distribution with mean $\delta$ and variance $\sigma^2$.
\section{System model and problem formulation}\label{sys_mod}
\subsection{Network Architecture and Channel Model}
Consider the downlink of a multi-Tags SR system consisting of a primary transmitter (PT) equipped with $M $ antennas, $K$ single-antenna passive
Tags, and a primary receiver (PR) equipped with $N$ antennas, as shown in Fig. \ref{fig:system_model}. For such scenario, the BL transmission reuses the spectrum of the DL
transmission.

\begin{figure}[h]
\centering
\includegraphics[width=9cm]{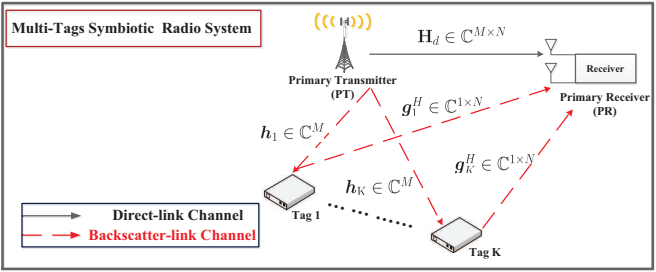}
\caption{An illustration of the downlink of a multi-Tags SR system.}\label{fig:system_model}
\end{figure}

Furthermore, each Tag $k$ modulates its information bits over the ambient RF signals by dynamically adjusting its antenna impedance, without requiring active RF components in
its transmission part, as illustrated in Fig. \ref{fig:passiveTag}. Specifically, Tag $k$ switches its impedance into the backscattered state, and backscatters the incoming
signal. Otherwise, Tag $k$ switches its impedance into the matched state, and dose not backscatter any signal \cite{binary_mod}. Hereinafter, we let $\alpha_k\in[0,1]$ be a
constant reflection coefficient\footnote{The refection coefficient of each Tag depends on the specific hardware implementation, such as the reradiation of closed-circuited
antenna, the chip impedance, and many others \cite{Cooperative_BSC,refCons}.} of Tag $k$. Moreover, we adopts the on-off keying (OOK) modulation scheme for BL transmission due
to the simple circuit design. Each Tag $k$ decides whether or not to transmit signal with probability $\rho_k$ in an i.i.d. manner. In this case, the symbol transmitted
from Tag $k$ with symbol period $T_b$ can be expressed as follows:
$$ b_{k}=\left\{
\begin{aligned}
 1,&~\mathrm{if~Tag}~k~\mathrm{is~in~backscattered~state},\\
 0,&~\mathrm{otherwise},
\end{aligned}
\right.
$$
so that $\mathrm{Pr}(b_k=1)=\rho_k$, and $\mathrm{Pr}(b_k=0)=1-\rho_k$.

Following the the earlier works on the SR scenario \cite{MA_BSC,Cooperative_BSC,ResourceA_BSC}, we also assume that the distance between any Tag and the PR is much shorter
than that between the PT and the PR, due to the limited transmission range of the backscatter communication. Several time synchronization methods have been proposed to
accurately estimate signal propagation delay, see e.g. \cite{SymRadio}, \cite{timeSyn}. By applying the above methods, all the propagation delay can be compensated properly
and the signals are synchronized perfectly.

\begin{figure}[!h]
\centering
\includegraphics[width=7cm]{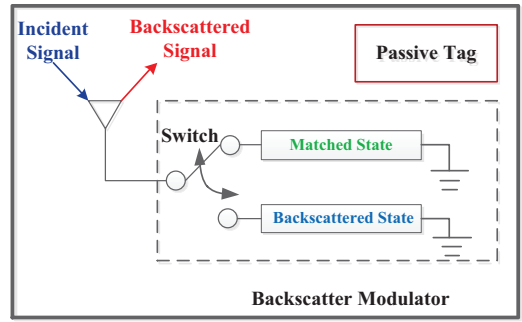}
\caption{Diagram of the analog signal chain for a passive Tag.}\label{fig:passiveTag}
\end{figure}

For clarity, we consider a block-fading channel model where the channel is static in each block. The DL channel from the PT to the PR is denoted by $\mathbf{H}_d\in
\mathbb{C}^{M\times N}$. Further, we let $\bm{h}_k\in \mathbb{C}^{M}$ and $\bm{g}^H_k \in \mathbb{C}^{1\times N}$ respectively denote the BL channel from the PT to Tag $k$, as
well as from Tag $k$ to the PR.
In this case, the signal received at the PR can be written as
\begin{equation}
\bm{y}(t)=\mathbf{H}_d^H\bm{v} s(t)+\sum_{k=1}^{K} \sqrt{\alpha_k }\bm{g}_k \bm{h}^H_k\bm{v}s(t)b_k(t)+\bm{w}(t),\nonumber
\end{equation}
where $\bm{v}\in \mathbb{C}^{M}$ is the transmit beamforming employed at the PT; $s(t)$ is the signal transmitted by the PT with symbol period $T_s$, and $\bm{w}(t)$ is the
additive white Gaussian noise (AWGN) at the PR.

\subsection{Timeline and Implementation Consideration}
In the SR regime, the symbol transmitted from each Tag $b_k(t)$ has a much longer symbol period than the symbol transmitted from the PT $s(t)$, i.e., $T_b=L T_s$ with an
integer constant $L\gg1$. The coherence time of channel is divided into timeslots and each timeslot consists of $L$ symbols, as illustrated in Fig. \ref{fig:timeStructure}.
Specifically, all the Tag symbols are assumed to be constant within the same timeslot.
\begin{figure}[h]
\centering
\includegraphics[width=9cm]{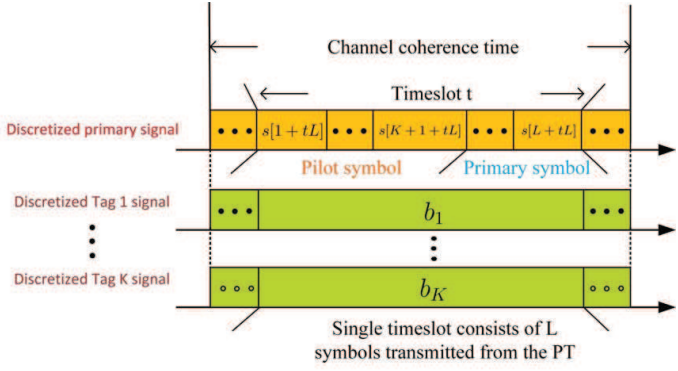}
\caption{An illustration of the timeline.}\label{fig:timeStructure}
\end{figure}

Within one particular timeslot, we discretize the received signal $\bm{y}(t)$ with the sampling rate $\frac{1}{T_s}$, and obtain the following expression with perfect
synchronization:
\begin{align}\label{received_signal}
    \bm{y}_l=\mathbf{H}^H_d \bm{v} s_l+\sum_{k=1}^{K}\sqrt{\alpha_k }b_k\bm{g}^H_k\bm{h}^H_k\bm{v}s_l+\bm{w}_l,
\end{align}
where $\bm{y}_l\triangleq \bm{y}(lT_s)$ is the $l$-th received signal at the PR, ${s}_l\triangleq {s}(lT_s)\sim\mathcal{CN}(0,1)$ is the $l$-th primary signal transmitted from
the PT, $b_k$ is the symbol transmitted from Tag $k$ within current timeslot, and $\bm{w}_l\triangleq \bm{w}(lT_s)$ is the AWGN with i.i.d. entries distributed as
$\mathcal{CN}(0,\sigma_w^2)$.

In this paper, we assume that the multi-Tags SR system operates in the time-division duplexing (TDD) mode \cite{TDD_mode}, and the channel state information (CSI) can be
obtained by uplink training. In particular, the uplink channel estimation in each particular block is divided into two phases. In both phases of training, the PR sends uplink
pilot signals and the PT estimates the channel. In the first phase, each Tag switches its impedance into the matched state, so that PT can estimate the uplink DL CSI based on
the received pilot signal. Then, the corresponding downlink DL CSI $\mathbf{H}_{\mathrm{d}}$ is obtained by exploiting the channel reciprocity.
In subsequent phase, each Tag switches its impedance into the backscattered state in turns, so that PT can estimate the corresponding BL effective CSI
$\tilde{\bm{h}}_k\triangleq b_k \hat{\bm{h}}_k \in\mathbb{C}^{M\times N}$ with $\hat{\bm{h}}_k\triangleq \bm{h}_k\bm{g}_k$ by subtracting the estimated DL CSI from the
estimated composite channel $\mathbf{H}_{d}+\tilde{\bm{h}}_k$  \cite{SymRadio}. As a result, both PT and PR can obtain CSI of the DL and BL transmission.

For convenience, we let $\hat{\bm{h}}=[\mathrm{vec}^T(\hat{\bm{h}}_1),\cdots,\mathrm{vec}^T(\hat{\bm{h}}_K)]^T\in \mathbb{C}^{KMN}$,
$\tilde{\bm{h}}=[\mathrm{vec}^T(\tilde{\bm{h}}_1),\cdots,\mathrm{vec}^T(\tilde{\bm{h}}_K)]^T\in \mathbb{C}^{KMN}$ , and $\bm{h}\triangleq
[\mathrm{vec}^T\left(\mathbf{H}_{d}\right),\hat{\bm{h}}^T]^T$.
Hence, the proposed stochastic transceiver is designed based on the effective BL CSI $\tilde{\bm{h}}$ and DL CSI $\mathbf{H}_{\mathrm{d}}$, rather than directly using the
respective knowledge of $\bm{h}_k$ and $\bm{g}_k$. Note that $\mathbf{H}_{d}$ only needs to be estimated once per block (channel coherence time), while $\tilde{\bm{h}}$ needs
to be estimated once per timeslot.
As observed in earlier works  \cite{MA_BSC,Cooperative_BSC,SymRadio,ResourceA_BSC}, the training overhead for channel estimation and time synchronization can be ignored due to
the limited number of Tags and large $L$.

\subsection{SINR Expression}
In this paper, we also assume that the signal strength of the backscatter link is much weaker than that of the direct link, due to the following facts. First, the hardware
capability of the Tags is limited, and thus the backscattering operation would inevitably result in the power loss of the Tags' symbols. Second, all the Tags' symbols are
transmitted through two channels and suffer from double attenuation. This is a typical operating regime that has been assumed in many works on SR systems
\cite{SymRadio,FD_BSC,Noma_BSC}. In this regime, the PR can apply the DLI cancellation technique to detect and decode the desired message based on the estimated CSI
$\mathbf{H}_d$ and $\tilde{\bm{h}}_k$, where the primary information $\hat{s}_l$ is first retrieved from the received signal and then the backscattered signal is obtained after eliminating the retrieved signal
$\hat{s}_l$.

\setcounter{equation}{1}

Hereinafter, we let $\bm{u}_s\in \mathbb{C}^N$ and $\bm{u}_k\in \mathbb{C}^N$ respectively denote the linear receive beamforming employed at the PR for detecting the primary
signal and Tag $k$ signal, and $\bm{u}\triangleq[\bm{u}^T_s,\bm{u}^T_1,\cdots,\bm{u}^T_K]\in\mathbb{C}^{N(K+1)}$. Applying the receive beamforming $\bm{u}_s$ at the PR, the
estimated signal is given by
\begin{align}
\hat{s}_l=\bm{u}^H_s\bm{y}_l= \bm{u}^H_s\mathbf{H}^H_d \bm{v} s_l+\sum_{k=1}^{K}\sqrt{\alpha_k }\bm{u}^H_s\tilde{\bm{h}}^H_k\bm{v}s_l+\bm{u}^H_s\bm{w}_l.\nonumber
\end{align}

For given optimization variable $\bm{\theta}\triangleq[\bm{v}^T,\bm{u}^T]^T$, Tags' symbol vector $\bm{b}\triangleq[b_1,\cdots,b_K]^T$, and channel realization $\bm{h}$,
the SINR of the DL transmission given by
\begin{equation}
\gamma_s(\bm{\theta};\bm{b})=\frac{|\bm{u}^H_s\bm{h}^H_{\mathrm{eq}}(\bm{b})\bm{v}|^2}{\sigma_w^2 \bm{u}^H_s \bm{u}_s},
\end{equation}
where $\bm{h}_{\mathrm{eq}}(\bm{b})=\mathbf{H}_d+\sum_{k=1}^K\sqrt{ \alpha}_k \tilde{\bm{h}}_k $ is the effective channel for decoding $s_l$.


In the subsequent phase, the PR eliminates the DLI based on the retrieved signal $\hat{s}_l$ and the estimated CSI $\mathbf{H}_d$. Consequently, it yields the following expression:
\begin{equation}
\bm{y}_{b,l}=\sum_{k=1}^{K}\sqrt{\alpha_k }\hat{\bm{h}}^H_k\bm{v}s_l b_k+\bm{w}_l.
\end{equation}

Then, the PR apply the linear receive beamforming vector $\bm{u}_k$ to obtain the estimated signal for Tag $k$ as follows
\begin{equation}
\hat{y}_{k,l}=\bm{u}^H_k \bm{y}_{b,l}=\sum_{m=1}^{K}\sqrt{\alpha_k } \bm{u}^H_k\hat{\bm{h}}^H_m\bm{v}s_l b_m+\bm{u}^H_k \bm{w}_l.
\end{equation}

Since each Tag only transmits a single symbol within one particular timeslot, the primary signal $s_l$ can be treated as the spread spectrum code of length $L$ when decoding
$b_k,\forall k$. By exploiting the temporary diversity via the despreading operation using the decoded primary signal $s_l$'s, the signal after despreading can be expressed as
\begin{align}
\hat{y}_k=\sum_{m=1}^K \bm{h}_{km} b_m +\bm{w}_{e}=\bm{h}_{kk}b_{k}+\sum_{m\neq k}\bm{h}_{km}b_{m}+\bm{w}_{e},
\end{align}
where $\bm{h}_{km}={\frac{1}{L}\sqrt{\alpha_m} \bm{u}^H_k\hat{\bm{h}}^H_m\bm{v}}\|\bm{s}\|$
with $\bm{s}\triangleq[s_1,\cdots,s_L]^T$ is the effective channel after despreading,
and $\bm{w}_e=\sum_{l=1}^L \frac{s^{\ast}_l\bm{u}^H_k \bm{w}_l}{L\|\bm{s}\|} $ is the effective noise.
Note that the beamforming vector $\bm{u}_{k}$ is designed to suppress the inter-tag-interference $\sum_{m\neq k}\bm{h}_{km}b_{m}$. As such, one can easily recover the original OOK symbol $b_{k}$ for tag k from $\hat{y}_{k}$.

Given the optimization variable $\bm{\theta}$, primary signal vector $\bm{s}$, and channel realization $\bm{h}$,
the SINR of the BL transmission from Tag $k$ is given by
\begin{equation}
{\gamma}_{k}(\bm{\theta};\bm{s})=\frac{\alpha_k \Lambda_{k} \|\bm{s}\|^2 |\bm{u}^H_k \hat{\bm{h}}^H_k \bm{v}|^2}{\sum_{m\neq k}^{K}\alpha_m \Lambda_{m} \|\bm{s}\|^2 |\bm{u}^H_k
\hat{\bm{h}}^H_m \bm{v}|^2 +\sigma^2_w \bm{u}^H_k \bm{u}_k},
\end{equation}
where $\Lambda_{k}\triangleq \rho_k(1-\rho_k)$.

Using the above notations, the expected SINR of the DL and the BL transmission are given by $\hat{\gamma}_s(\bm{\theta})=\mathbb{E}_{\bm{b}}[\gamma_s(\bm{\theta};\bm{b})]$ and
$\hat{\bm{\gamma}}(\bm{\theta})=[\hat{\gamma}_{1}(\bm{\theta}),\cdots,\hat{\gamma}_{K}(\bm{\theta})]^T$ with $\hat{\gamma}_{k}(\bm{\theta})=
\mathbb{E}_{\bm{s}}[{\gamma}_{k}\left(\bm{\theta};\bm{s}\right)]$, respectively. For convenience, we let
$\hat{\bm{\gamma}}_A\triangleq \big[\hat{\gamma}_{s}(\bm{\theta});\hat{\bm{\gamma}}(\bm{\theta})\big]\nonumber
$
denote the composite average SINR vector.

\begin{remark}
In practice, it may be difficult to perform perfect DLI cancellation due to potential sources of errors such as channel estimation errors and decoding errors. In this case, the received signal after DLI cancellation can be expressed as
\begin{equation}
\bm{y}_{b,l}=\sum_{k=1}^{K}\sqrt{\alpha_{k}}\hat{\bm{h}}_{k}^{H}\bm{v}s_{l}b_{k}+\sqrt{\varrho}\mathbf{H}_{d}^{H}\bm{v}s_{l}+\bm{w}_{l},\nonumber
\end{equation}
where $0<\varrho<1$ denotes the power coefficient of the residual interference after DLI cancellation. Following the earlier work \cite{IoT_Challenge}, the residual error after DLI cancellation can be approximated as Gaussian noise, thereby increasing the noise variance of the whole system. The proposed BSPD algorithm can be readily extended to such a scenario with imperfect DLI cancellation. However, the concrete model about the residual error after DLI cancellation is beyond the scope of this paper.
\end{remark}


\subsection{Problem formulation}
The joint optimization of the transmit and receive beamforming can be formulated as the general network utility maximization problem (GUMP):
\begin{align}
\mathcal{P}_G: &\mathop{\max}_{\bm{\theta}\in \bm{\Omega}}~ f\left(\hat{\bm{\gamma}}_A\right), \end{align}
where $\bm{\Omega}\triangleq\{\bm{v}|\mathrm{Tr}\left(\bm{v}\bm{v}^H\right)\leq P_{\mathrm{max}}\}$ is the feasible set of optimization variables $\bm{\theta}$, and
$P_{\mathrm{max}}$ is the power budget of the PT. The network utility function $f(\cdot)$ is a continuously differentiable and concave function of argument. Moreover,
$f(\cdot)$ is non-decreasing with respect to each component of argument and its corresponding partial derivative is Lipschitz continuous. Example of commonly used utility
functions are the sum performance utility, proportional fairness utility, and harmonic mean rate utility\cite{utility,XihanTVT,XihanIoT}.
For clarity, we focus on a novel SINR utility function as an example, which is designed based on the barrier method \cite{NP_OPT} such that by maximizing it, the PR's expected SINR can be maximized and meanwhile, each Tag's each Tag's expected SINR can be guaranteed to exceed certain threshold (with high accuracy).
As such, the optimization problem $\mathcal{P}_G$ becomes:
\begin{align}\label{appGump}
\mathcal{P}: &\mathop{\max}_{\bm{\theta}\in \bm{\Omega}}~ f\left(\hat{\bm{\gamma}}_A\right)=\hat{\gamma}_s(\bm{\theta})\!+\!\frac{1}{\psi}\sum_{k=1}^K
\log(\hat{\gamma}_k(\bm{\theta})-\gamma_{0,k}),
 \end{align}
where $\gamma_{0,k}$ is a positive average SINR target for the $k$-th BL transmission link, $\frac{1}{\psi}\sum_{k=1}^{K}\log(\hat{\gamma}_{k}(\bm{\theta})-\gamma_{0,k})$ is the barrier function used to guarantee the SINR requirement (i.e., guarantee that the expected SINR of all BL transmission links is no less than the corresponding targets $\gamma_{0,k}$'s) and $\psi$ is a parameter to control the price of the barrier function. Here, we remark that the above utility is different from the traditional weighted sum performance utility \cite{utility} in that the barrier function in \eqref{appGump} can guarantee the SINR requirement for all Tags with high accuracy, since $\log(\hat{\gamma}_{k}(\bm{\theta})-\gamma_{0,k})$ will tend to minus infinity as the SINR $\hat{\gamma}_{k}(\bm{\theta})$ deviates from the target $\gamma_{0,k}$. Specifically, the parameter $\psi$ can be used to control the tradeoff between the accuracy of the SINR requirement guarantee and the smoothness of the barrier function.

Note that there are several challenges in finding stationary solutions of problem $\mathcal{P}$, elaborated as follows. First, the objective function contains expectation
operators, which is difficult to have a closed-form expression. Second, problem  $\mathcal{P}$ is typically a function of multiple-ratio FP, where the optimization variable
$\bm{\theta}$ appears in both numerator and the denominator of $\hat{\bm{\gamma}}_A$. In particular, solving the above multiple-ratio FP is always NP-hard.
To the best of our knowledge, there lacks an efficient algorithm to handle such stochastic nonconvex optimization problem $\mathcal{P}$.

\section{Stochastic Parallel Decomposition Algorithm for General Utility Optimization}\label{EGUO}
In this section, we propose a BSPD algorithm to find the stationary solutions of problem $\mathcal{P}$. We shall first define the stationary solutions of problem
$\mathcal{P}$. Then we elaborate the implementation details of the proposed algorithm, and establish its local convergence.
\subsection{Stationary Point of Problem $\mathcal{P}$}
Before elaborating the implementation details of the proposed algorithm, the stationary solutions of problem $\mathcal{P}$ is defined as:
\begin{definition}\label{def1}
A solution $\bm{\theta}^{\star}$ is called a stationary solution of problem $\mathcal{P}$ if it satisfies the following condition:
\begin{equation}\label{stationaryCon}
 (\bm{\theta}-\bm{\theta}^{\star})^T \mathbf{J}_{\bm{\theta}}(\bm{\theta}^{\star})\nabla_{\hat{{\bm{\gamma}}}_A} f\left(\hat{{\bm{\gamma}}}^{\star}_A\right)\leq 0,\forall \bm{\theta}
 \in \mathbf{\Omega},
\end{equation}
where $\mathbf{J}_{\bm{\theta}}(\bm{\theta}^{\star})$ is the Jacobian matrix \footnote{The Jacobian matrix of $\hat{\bm{\gamma}}_A(\bm{\theta})$ is defined as
$\mathbf{J}_{\gamma}(\bm{\theta})=[\nabla_{\bm{\theta}} \hat{\gamma}_s,\nabla_{\bm{\theta}} \hat{\gamma}_1,\cdots,\nabla_{\bm{\theta}} \hat{\gamma}_K]$, where $\nabla_{\bm{\theta}} \hat{\gamma}_k$ is
the partial derivative of $\hat{\gamma}_k$ with respect to $\bm{\theta}$.}
of the SINR vector $\hat{{\bm{\gamma}}}_A$ with respect to $\bm{\theta}$ at $\bm{\theta}=\bm{\theta}^{\star}$, and $\nabla_{\hat{\bm{\gamma}}_A} f\left(\hat{{\bm{\gamma}}}^{\star}_A\right)$ is
the derivative of $f\left(\hat{{\bm{\gamma}}}_A\right)$ at $\hat{{\bm{\gamma}}}_A=\hat{{\bm{\gamma}}}^{\star}_A\triangleq \hat{{\bm{\gamma}}}_A(\bm{\theta}^{\star})$.
\end{definition}

In the following, we propose a BSPD algorithm to find stationary points of problem $\mathcal{P}$.


%


\subsection{Proposed BSPD Algorithm for Solving Problem $\eqref{appGump}$}

The proposed BSPD algorithm is summarized in Algorithm \ref{alg:BSPD}.
In BSPD algorithm, an auxiliary weight vector $\bm{\nu}\triangleq[\nu_0,\cdots,\nu_K]^T$ is introduced to approximate the derivative $\nabla_{\hat{\bm{\gamma}}_A}
f\left(\hat{\bm{\gamma}}_A\right)$. Note that traditional SAA method needs to collect a large number of samples for the random state $\{\bm{s},\bm{b}\}$ before solving the stochastic
optimization problem \eqref{appGump} and processes all samples per iterations. Hence,  it requires huge memory to store the samples and the corresponding computational
complexity is also higher. To address the above issues, we use mini-batch method \cite{miniBatch} to achieve good convergence speed with the reduced computational complexity.
In each iteration, $J$ mini-batches are generated for $\bm{s}$ and $\bm{b}$ independent and identically drawn from the primary data distribution and tag data distribution
specified in Section \ref{sys_mod}, respectively. Let superscript $t$ denote variables associated with the $t$-th iteration. In the $t$-th iteration, one random mini-batch
$\{\bm{s}^t_j, \bm{b}^t_j,j=1,\cdots,J\}$ of size $J$ is generated. Here, we only briefly discuss the impact of $J$ on the overall convergence of Algorithm \ref{alg:BSPD}. A
larger $J$
usually leads to a faster overall convergence speed at the cost of higher computational complexity at each iteration. When $J=1$, the BSPD algorithm has the lowest complexity per iteration, but the total number of iterations will also increase.

Note that the optimal $\bm{\theta}$ and $\bm{\nu}$ cannot be obtained by directly maximizing $f\left(\hat{\bm{\gamma}}_A\right)$ because $f\left(\hat{\bm{\gamma}}_A\right)$ is not
concave and it does not have closed-form expression. To enable the original nonconvex problem amenable to optimization, we first recursively approximate the average SINR
vector $\hat{\bm{\gamma}}_A$ based on the online observations of random mini-batch, which is given by
\begin{align}
\tilde{\gamma}^t_s&=(1-\xi^t)\tilde{\gamma}^{t-1}_s+\xi^t\sum_{j=1}^J \gamma_s(\bm{\theta}^{t-1},\bm{b}^{t}_j)/J,\label{appRs}\\
\tilde{\gamma}^t_k&=(1-\xi^t)\tilde{\gamma}^{t-1}_k+\xi^t\sum_{j=1}^J  \gamma_k(\bm{\theta}^{t-1},\bm{s}^{t}_j)/J,\label{appRk}
\end{align}
with $\tilde{\gamma}^{0}_s=\tilde{\gamma}^{0}_k={0}$, where $\xi^t\in(0,1]$ is a step-size sequence to be properly chosen.
Hereinafter, we let $\tilde{\bm{\gamma}}^t_A=[\tilde{\gamma}^t_s,\tilde{\gamma}^t_1,...,\tilde{\gamma}^t_K]^T$ denote the recursive approximation of the average SINR vector $\hat{\bm{\gamma}}_A$ in
the $t$-th iteration.

Then the weight vector $\bm{\nu}$ is updated as:
\begin{equation}\label{update_nu}
    \bm{\nu}^{t}=(1-\omega^t)\bm{\nu}^{t-1}+\omega^t \tilde{\bm{\nu}}^t,
\end{equation}
where $\omega^t\in(0,1]$ is a step-size sequence satisfying $\sum_t \omega^t=\infty$, $\sum_t (\omega^t)^2<\infty$, and $\tilde{\bm{\nu}}^t\triangleq \nabla_{\hat{\bm{\gamma}}_A} f(\tilde{\bm{\gamma}}^t_A)$.

Based on the recursive approximation $\tilde{\bm{\gamma}}^t_A$, the recursive approximation of the partial derivative $\nabla_{\bm{\theta}}f(\hat{\bm{\gamma}}_A)$ with respect to
$\bm{\theta}$ is given by
\begin{align}
\bm{\Xi}^t&=\mathbf{F}^t \bm{\nu}^t,\label{appGrad_1}\\
\mathbf{F}^t&=(1-\xi^t)\mathbf{F}^{t-1}+\xi^t \nabla_{\bm{\theta}}\tilde{\bm{\gamma}}^t_A,\label{appGrad_2}
\end{align}
along with $\mathbf{F}^0=\bm{0}$. It will be shown in Lemma \ref{lemma_2} that the recursive approximation $\bm{\gamma}^t_A$ and $\bm{\Xi}^t$ respectively converge to true average
SINR and partial derivative, which address the issue of no closed-form characterization of the average SINR vector and guarantees the convergence of Algorithm
\ref{alg:BSPD}.

\begin{algorithm}[t]

\caption{\label{alg:BSPD}Proposed BSPD Algorithm for solving $\mathcal{P}$}

\textbf{\small{}Input: }{\small{}Step-size sequence $\{\xi^t,\omega^t\}$.}{\small\par}

\textbf{\small{}Initialization:}{\small{} $\bm{\theta}^0\in\bm{\Omega}$; $\bm{\nu}^0=[1,\cdots,1]^T$,
and $t=1$.}{\small\par}

\textbf{Repeat}

\textbf{\small{}\,\,\,\,\,\,\,\,\,\,Step 1: }{\small{}Realize $J$ mini-batches of independent random states $\\~~~~~~~~~~~~~~~$$\{\bm{s}^t_j, \bm{b}^t_j\}$.}

\textbf{\small{}\,\,\,\,\,\,\,\,\,\,Step 2: }{\small{}Calculate $\tilde{\bm{\nu}}^t= \nabla_{\hat{\bm{\gamma}}_A} f(\tilde{\bm{\gamma}}^t_A)$ and update $\bm{\nu}^{t}$ according
$\\~~~~~~~~~~~~~~~$ to \eqref{update_nu}.}

\textbf{\small{}\,\,\,\,\,\,\,\,\,\,Step 3: }{\small{}Update the surrogate function $\tilde{f}^t(\bm{\theta})$ according to \eqref{surrogate}.}

\textbf{\small{}\,\,\,\,\,\,\,\,\,\,Step 4: }{\small{}Apply Algorithm 2 with input $\{\bm{s}^t_j,\bm{b}^t_j\}$, $\xi^t$ $\bm{\nu}^{t}$ and
$\\~~~~~~~~~~~~~~~~$$\bm{\theta}^{t-1}$ to obtain $\tilde{\bm{\theta}}^t$, as elaborated in Section \ref{sec_Short}.}

\textbf{\small{}\,\,\,\,\,\,\,\,\,\,Step 5: }{\small{}Update $\bm{\theta}^{t}$ according to \eqref{update_theta}.}

\textbf{until }the value of objective function $\tilde{f}^t(\bm{\theta})$
converges. Otherwise, let $t\leftarrow t+1$.
\end{algorithm}

Based on the generated mini-batch $\{\bm{s}^t_j, \bm{b}^t_j\}$, the current iterate $\bm{\nu}^{t}$ and the last iterate $\bm{\theta}^{t-1}\triangleq [(\bm{v}^{t-1})^T,(\bm{u}^{t-1})^T]^T$,
we tailor a surrogate function of the objective with the specific structure as follows:
\begin{align}\label{surrogate}
\tilde{f}^t(\bm{\theta})=&\xi^t\tilde{g}^t(\bm{\theta})+(1-\xi^t)(\bm{\Xi}^{t})^T(\bm{\theta}-\bm{\theta}^{t-1})-q(\bm{\theta},\bm{\theta}^{t-1}),
\end{align}
where
the proximal regularization function $q(\bm{\theta},\bm{\theta}^{t-1})$ is  given by
\begin{align}
q(\bm{\theta},\bm{\theta}^{t-1})=&\tau_{\bm{v}}\|\bm{v}-\bm{v}^{t-1}\|^2+\tau_{\bm{u}}\|\bm{u}-\bm{u}^{t-1}\|^2,
\end{align}
with $\tau_{\bm{v}},\tau_{\bm{u}}>0$ are positive constants; and $\tilde{g}^t(\bm{\theta})$ is the sample average approximation of original objective function
$f(\hat{\bm{\gamma}}_A)$,
given by
\begin{align}
\tilde{g}^t(\bm{\theta})=&\frac{1}{J}\sum_{j=1}^J \big(g^t_j(\bm{v},\bm{u}^{t-1})+g^t_j(\bm{v}^{t-1},\bm{u})\nonumber\\
&-g^t_j(\bm{v}^{t-1},\bm{u}^{t-1})\big),
\label{BotE3_1}
\end{align}
where $g^t_j(\bm{v},\bm{u})$ is the approximate objective function for the specific realization of random system states $\{\bm{s}^t_j,\bm{b}^t_j\}$ given by
\begin{align}
g^t_j(\bm{v},\bm{u})=&  \nu_0 \overline{\gamma}^t_s\big(\bm{v},\bm{u},\bm{b}^t_j)
+\sum_{k=1}^K \nu_k \overline{\gamma}^t_k(\bm{v},\bm{u},\bm{s}^t_j),
\end{align}
where the first term $\overline{\gamma}^t_s\big(\bm{v},\bm{u},\bm{b}^t_j)$ can be expressed as
\begin{align}
\overline{\gamma}^t_s\big(\bm{v},\bm{u},\bm{b}^t_j)=&2\mathrm{Re}\{(\phi^t_s(\bm{b}^t_j))^H \bm{u}^H_s\bm{h}_{\mathrm{eq}}^H\bm{v}\}\nonumber\\
&-\sigma^2_w
(\phi^t_s(\bm{b}^t_j))^H  \bm{u}^H_s \bm{u}_s\phi^t_s(\bm{b}^t_j)
\end{align}
where $\phi^t_s(\bm{b}^t_j)$ is a constant  depending on the last iterate $(\bm{v}^{t-1},\bm{u}^{t-1})$ and the current realization of
$\bm{b}^t_j$ given by
\begin{align}\label{botE2_1}
\phi^t_s(\bm{b}^t_j)=(\sigma^2_w\bm{u}^H_s \bm{u}_s)^{-1}\bm{u}^H_s\bm{h}_{\mathrm{eq}}^H\bm{v},
\end{align}
and the second term $\overline{\gamma}^t_k\big(\bm{v},\bm{u},\bm{s}^t_j)$ can be expressed as
\begin{align}
\overline{\gamma}^t_k\big(\bm{v},\bm{u},\bm{s}^t_j)=&2\mathrm{Re}\{\sqrt{\alpha_k\|\bm{s}^t_j\|^2\Lambda_k}(\phi^t_k(\bm{s}^t_j))^H\bm{u}^H_k\hat{\bm{h}}^H_k \bm{v}\}\nonumber\\
&-(\phi^t_k(\bm{s}^t_j))^H{\Gamma}_k(\bm{v},\bm{u},\bm{s}) \phi^t_k(\bm{s}^t_j),
\end{align}
where  ${\Gamma}_k(\bm{v},\bm{u},\bm{s}^t_j)$ is the function of $(\bm{v},\bm{u})$ given by
\begin{equation}
{\Gamma}_k(\bm{v},\bm{u},\bm{s}^t_j)=\sum_{m\neq k}^{K}\alpha_m \Lambda_{m} \|\bm{s}^t_j\|^2 |\bm{u}^H_k
\hat{\bm{h}}^H_m \bm{v}|^2 +\sigma^2_w \bm{u}^H_k \bm{u}_k,
\end{equation}
and  $\phi^t_k(\bm{s}^t_j)$ is a constant depending on the last iterate $(\bm{v}^{t-1},\bm{u}^{t-1})$ and the current realization of
$\bm{s}^t_j$ given by
\begin{equation}\label{botE2_3}
\phi^t_k(\bm{s}^t_j)={\Gamma}^{-1}_k(\bm{v},\bm{u},\bm{s}^t_j)\sqrt{\alpha_k\|\bm{s}^t_j\|^2\Lambda_k}(\phi^t_k(\bm{s}^t_j))^H\bm{u}^H_k\hat{\bm{h}}^H_k \bm{v}).
\end{equation}


Using the above notations, we aim to prove the following lemma, which plays a key role in the establishing the local convergence of the proposed BSPD algorithm to stationary solutions.
\begin{lemma}\label{consisLemma}
Both the  gradient and the objective value
of the approximate function $\tilde{g}^t(\bm{\theta})$ are unbiased estimation of the weighted sum average SINR, i.e.,
\begin{align}
\overline{\gamma}^t_s\big(\bm{v}^{t},\bm{u}^{t}\bm{b}^t_j)&=\gamma_s(\bm{\theta}^{t},\bm{b}^{t}_j),\label{consisT_1}\\
\overline{\gamma}^t_k\big(\bm{v}^{t},\bm{u}^{t},\bm{s}^t_j)&=\gamma_k(\bm{\theta}^{t},\bm{s}^{t}_j),\label{consisT_2}\\
\mathbb{E}_{\bm{b,s}}\left[g^t_j(\bm{v}^t,\bm{u}^t)-(\bm{\nu}^t)^T\hat{\bm{\gamma}}_A(\bm{\theta}^t)\right]&=0,\label{consisT_3}\\
\mathbb{E}_{\bm{b,s}}\left[\nabla_{\bm{\theta}}g^t_j(\bm{v}^t,\bm{u}^t)-\nabla_{\bm{\theta}}(\bm{\nu}^t)^T\hat{\bm{\gamma}}_A(\bm{\theta}^t)\right]&=0.\label{consisT_4}
\end{align}
%
\end{lemma}

\begin{proof}
By exploiting the quadratic transform method \cite{FP_part2,FP_part1}, $\overline{\gamma}^t_s\big(\bm{v},\bm{u},\bm{b}^t_j)$ and
$\overline{\gamma}^t_k\big(\bm{v},\bm{u},\bm{s}^t_j)$ respectively follow from ${\gamma}_s(\bm{\theta}^{t},\bm{b}^{t}_j)$ and ${\gamma}_k(\bm{\theta}^{t},\bm{s}^{t}_j)$. Consequently, we have
\eqref{consisT_1} and \eqref{consisT_2}. Combining \eqref{consisT_1} and \eqref{consisT_2} with the fact that the random state $\{\bm{s},\bm{b}\}$ is bounded and identically distributed, both \eqref{consisT_3} and \eqref{consisT_4} immediately hold.  This completes the
proof.
\end{proof}
%

For fixed $\bm{\nu}$, we first obtain the optimal solution $\tilde{\bm{\theta}}^t$ of the following problem:
\begin{align}\label{surroP}
\mathop{\max}_{\bm{\theta}\in \bm{\Omega}}~\tilde{f}^t(\bm{\theta}).
\end{align}

\setcounter{equation}{32}

\begin{figure*}[b]
\hrulefill
\begin{align}
    \bm{v}^{t,\star}(\mu)=&\big(\mathbf{\Upsilon}^t(\mu)\big)^{-1} \Big(\frac{\xi^t}{J}\sum_{j=1}^J\Big[\sum^K_{k=1}
    \nu_k \sqrt{\alpha_k\Lambda_k\|\bm{s}^t_j\|^2}\phi_k(\bm{\theta}^{t-1},\bm{s}^t_j)\hat{\bm{h}}_k\bm{u}_k+
   \nu_0 \phi_s(\bm{\theta}^{t-1},\bm{b}^t_j) \bm{h}_{\mathrm{eq}}\bm{u}_s
    \Big]+(1-\xi^t)\nabla_{\bm{v}} \bm{J}^{t}_v+\tau_{\bm{v}}\bm{v}^{t-1} \Big).\label{opt_v}\\
    \mathbf{\Upsilon}^t(\mu)=&\frac{\xi^t}{J}\sum_{j=1}^J
    \sum^K_{k=1} \nu_k \sum^K_{m\neq k}\|\bm{s}^t_j\|^2 \alpha_m \Lambda_m{\phi}_k(\bm{\theta}^{t-1},\bm{s}^t_j) {\phi}^H_k(\bm{\theta}^{t-1},\bm{s}^t_j)  \hat{\bm{h}}_m
\bm{u}_k \bm{u}^H_k \hat{\bm{h}}^H_m
   +(\mu+\tau_{\bm{v}}) \mathbf{I}_M.\label{topE6_2}
\end{align}

\end{figure*}

\setcounter{equation}{26}

\setcounter{equation}{37}
\begin{figure*}[b]
\hrulefill
\begin{align}
    \bm{u}_s^{t,\star}=&\Big(\Big[\frac{\xi^t}{J}\sum_{j=1}^J\sigma^2_w \phi_s(\bm{\theta}^{t-1},\bm{b}^t_j)\phi^H_s(\bm{\theta}^{t-1},\bm{b}^t_j)
    +\tau_{\bm{u}}\Big]\mathbf{I}_N\Big)^{-1}
    \Big(\frac{\xi^t}{J}\sum_{j=1}^J
    \phi^H_s(\bm{\theta}^{t-1},\bm{b}^t_j)\bm{h}^H_{\mathrm{eq}}
    \bm{v}
    %
    +(1-\xi^t)\bm{J}^{t}_{u_s}+\tau_{\bm{u}}\bm{u}^{t-1}_s\Big)
    .\label{optimalUs}\\
    \bm{u}_k^{t,\star}=&\Big(\frac{\xi^t}{J}\sum_{j=1}^J \phi_k(\bm{\theta}^{t-1},\bm{s}^t_j)\phi^H_k(\bm{\theta}^{t-1},\bm{s}^t_j)\big(\sum_{m\neq k}^K\|\bm{s}^t_j\|^2 \alpha_k
    \Lambda_m \hat{\bm{h}}^H_{m}\bm{v}\bm{v}^H \hat{\bm{h}}_m +\sigma^2_w\mathbf{I}_N\big)+\tau_{\bm{u}}\mathbf{I}_N\Big)^{-1}\nonumber\\
    &\times \Big(\frac{\xi^t}{J}\sum_{j=1}^J \sqrt{\alpha_k\|\bm{s}^t_j\|^2\Lambda_k}\phi^H_k(\bm{\theta}^{t-1},\bm{s}^t_j)
    \hat{\bm{h}}^H_k \bm{v}
    +(1-\xi^t)\bm{J}^{t}_{u_k}+\tau_{\bm{u}}\bm{u}^{t-1}_k\Big).\label{optimalUk}
\end{align}
\end{figure*}
\setcounter{equation}{28}

For the detailed procedure for solving problem \eqref{surroP} in a parallel fashion,
it will be postponed to Section \ref{sec_Short}.
Moreover, $\bm{\theta}$ is updated according to
\begin{equation}\label{update_theta}
\bm{\theta}^t=(1-\omega^t)\bm{\theta}^{t-1}+\omega^t\tilde{\bm{\theta}}^t.
\end{equation}

Finally, the above steps are repeated until convergence.

\begin{remark}\label{advSurrogate}
 The potential advantages of the structured surrogate function for problem $\mathcal{P}$ are elaborated below. First, the numerator and denominator in \eqref{appGump} are now
 decoupled in problem \eqref{surroP}, which overcomes the difficulty from the nonlinear fractional coupling. Second, the structured surrogate function in problem
 \eqref{surroP} preserves the structure of the original problem, which helps to speed up the  convergence speed. Third, we also observe that the constraint is separable with
 respect to the two blocks of variables, i.e.,  $\bm{v}$, and $\bm{u}$. Therefore, we can decompose the original problem \eqref{surroP} into two independent subproblems with
 respect to each block of variables, where each subproblem is strongly convex and can be efficiently solved in a \emph{parallel and distributed} fashion, as detailed in Section \ref{sec_Short}.
\end{remark}

\begin{remark}
The reason that the optimization variable $\bm{\theta}^{t}$ are obtained by solving problem \eqref{surroP} is as follows. According to \eqref{stationaryCon}, it implies that a
stationary solution $\bm{\theta}^{\star}$ must be a stationary point of problem \eqref{surroP} with a \emph{stationary weight vector} $\bm{\nu}^{\star}=\nabla_{\hat{\bm{\gamma}}_A}
f\left(\hat{\bm{\gamma}}^{\star}_A\right)$, while $\bm{\nu}^{\star}$ is not known a prior. Consequently, the key idea of the proposed algorithm is to iteratively update the
optimization variable $\bm{\nu}^t$ until it converges to the corresponding stationary solution $\bm{\nu}^{\star}$. Then the limiting optimization variable
$\bm{\theta}^{\star}$ satisfying \eqref{stationaryCon} can be obtained by finding a stationary point of the corresponding  problem $\eqref{surroP}$ as $t\rightarrow \infty$.
\end{remark}

\begin{algorithm}[t]

\caption{\label{alg:shortAlg}Proposed Parallel Algorithm for solving \eqref{surroP}}

\textbf{\small{}Input: }{\small{}Last iterate $\bm{\theta}^{t-1}$ , current iterate $\bm{\nu}^{t}$, current step size $\xi^t$, and
current mini-batches $\{\bm{s}^t_j,\bm{b}^t_j\}$.
}{\small\par}

\textbf{\small{}Step 1: }{\small{}Obtain $\bm{v}^{t,\star}$ by solving problem \eqref{WSRMP_v}}.

\textbf{\small{}Step 2: }{\small{}Obtain $\bm{u}^{t,\star}$ by solving problem \eqref{subproblemU}}.

\textbf{\small{}Step 3: }{\small{}Terminate the algorithm and output $\{\bm{v}^{t,\star},\bm{u}^{t,\star}\}$}.
\end{algorithm}

\subsection{Proposed Parallel Optimization Algorithm for Solving Problem \eqref{surroP}}\label{sec_Short}
The proposed parallel optimization algorithm for solving problem \eqref{surroP} is summarized in Algorithm \ref{alg:shortAlg}. Here, we emphasize that Step 2 and 3 in
Algorithm \ref{alg:shortAlg} are executed in a parallel fashion.
At iteration $t$, Algorithm \ref{alg:shortAlg} is initialized with the input $\{\bm{s}^t_j,\bm{b}^t_j\}$, $\xi^t$, $\bm{\nu}^{t}$ and $\bm{\theta}^{t-1}$. Based on the
acquired input $\{\bm{s}^t_j,\bm{b}^t_j\}$, $\bm{\nu}^{t}$ and $\bm{\theta}^{t-1}$, the auxiliary variable $\bm{\phi}^t$  can be  calculated according to
\eqref{botE2_1} and \eqref{botE2_3}, where $\bm{\phi}^{t}=[(\bm{\phi}^t_s)^T,(\bm{\phi}^t_1)^T,\cdots,(\bm{\phi}^t_K)^T]^T$ with
$\bm{\phi}^t_s=[\phi^t_s(\bm{b}^t_1),\cdots,\phi^t_s(\bm{b}^t_J)]^T$ and $\bm{\phi}^t_k=[\phi^t_k(\bm{s}^t_1),\cdots,\phi^t_k(\bm{s}^t_J)]^T$.

Given auxiliary variable  $\bm{\phi}^t$, problem \eqref{surroP} can be decomposed into two independent strongly convex subproblems with
respect to each variable, which can be efficiently solved in a parallel and distributed manner. In the following, we show how these two subproblems are addressed.

\subsubsection{Optimization of $\bm{v}$}
We here focus on optimizing the transmit beamforming $\bm{v}$ by considering the following optimization problem
\begin{subequations}\label{WSRMP_v}
\begin{align}
&\mathop{\max}_{\bm{v}}~ \tilde{g}^t(\bm{v}) \\
    &~\textrm{s.t.} \quad \mathrm{Tr}\left(\bm{v}\bm{v}^H\right)\leq P_{\mathrm{max}},
\end{align}
\end{subequations}
where the objective function in \eqref{WSRMP_v} is given by
\begin{align}
 \tilde{g}^t(\bm{v})=&\frac{\xi^t}{J}\sum_{j=1}^J g^t_j(\bm{v},\bm{u}^{t-1})
+(1-\xi^t) (\bm{J}^{t}_v)^T(\bm{v}-\bm{v}^{t-1})\nonumber\\
&-\tau_{\bm{v}}\|\bm{v}-\bm{v}^{t-1}\|^2,
\end{align}
and $\bm{J}^{t}_v$ is the recursive approximation of the partial derivative $\nabla_{\bm{v}} f(\hat{\bm{\gamma}}_A)$ with respect to $\bm{v}$ as specified in
\eqref{appGrad_1}-\eqref{appGrad_2}.

Clearly, problem \eqref{WSRMP_v} is an quadratically constrained quadratic optimization problem, which can be solved by dealing with its dual problem. by introducing Lagrange
multiplier $\mu$ for the corresponding constraint $\mathrm{Tr}\left(\bm{v}\bm{v}^H\right)\leq P_{\mathrm{max}}$, we define the Lagrangian associated with problem
\eqref{WSRMP_v} as
\begin{align}
    \mathcal{L}^t(\bm{v},\mu)\triangleq  \tilde{g}^t(\bm{v})-\mu\left(\mathrm{Tr}\left(\bm{v}\bm{v}^H\right)- P_{\mathrm{max}}\right).
\end{align}

Since $\mathcal{L}^t(\bm{v},\mu)$ with respect to $\bm{v}$ for fixed Lagrange multiplier is an unconstrained quadratic optimization problem, it can be efficiently solved by
checking its first-order optimal condition, which yields the optimal $\bm{v}^{t,\star}(\mu)$ with $\mathbf{\Upsilon}^t(\mu)$ defined in \eqref{opt_v}-\eqref{topE6_2} as
displayed at the bottom of the next page. Note that $\mu$ in \eqref{opt_v} is chosen to be zero $\|\bm{v}^t(0)\|^2\leq P_{\mathrm{max}}$ and chosen to satisfy
$\|\bm{v}^t(\mu)\|^2= P_{\mathrm{max}}$ otherwise.

\subsubsection{Optimization of $\bm{u}$}
The subproblem with respect to receive beamforming $\bm{u}$ is given by
\setcounter{equation}{35}
\begin{align}\label{subproblemU}
&\mathop{\max}_{\bm{u}}~ \tilde{g}^t(\bm{u})
\end{align}
where the objective function in \eqref{subproblemU} is given by
\begin{align}
\tilde{g}^t(\bm{u})=&\frac{\xi^t}{J}\sum_{j=1}^J g^t_j(\bm{v}^{t-1},\bm{u})
+(1-\xi^t) (\bm{J}^{t}_u)^T(\bm{u}-\bm{u}^{t-1})\nonumber\\
&-\tau_{\bm{u}}\|\bm{u}-\bm{u}^{t-1}\|^2,
\end{align}
and $\bm{J}^{t}_u$ is the recursive approximation of the partial derivative $\nabla_{\bm{u}} f(\hat{\bm{\gamma}}_A)$ with respect to $\bm{u}$ as specified in
\eqref{appGrad_1}-\eqref{appGrad_2}.

It is seen that problem \eqref{subproblemU} is an unconstrained quadratic optimization, which can be efficiently solved by applying the first-order optimal condition. After
some calculations and appropriate rearrangement, the optimal $\bm{u}^{t,\star}_s$ and $\bm{u}^{t,\star}_k$ are defined in \eqref{optimalUs}-\eqref{optimalUk} as displayed at
the bottom of this page.

\subsection{Convergence and Computational Complexity}
In this subsection, we establish the local convergence of the proposed BSPD algorithm to stationary solutions. To this end, the step-size sequence $\{\xi^t,\omega^t\}$ needs
to satisfy the following conditions.
\begin{assumption}(Assumption on step sizes)\label{assumption1}
\begin{enumerate}[(1)]
\item $\xi^t \rightarrow 0, \sum_t \xi^t=\infty$, $\sum_t (\xi^t)^2<\infty$, $\lim_{t\rightarrow\infty} \xi^t t^{-\frac{1}{2}}<\infty$.
\item $\omega^t \rightarrow 0, \sum_t \omega^t=\infty$, $\sum_t (\omega^t)^2<\infty$.
\item $\lim_{t \rightarrow \infty} \omega^t/\xi^t=0$.
\end{enumerate}
\end{assumption}
\setcounter{equation}{39}

The motivation for the above conditions on $\{\xi^{t},\omega^{t}\}$ is explained below. First, Assumption 1-(1) and 1-(2) state that both $\xi^{t}$ and $\omega^{t}$ follow the diminishing step-size rule, while not approaching to zero too fast. Note that similar diminishing step-size rule has been assumed in many other stochastic optimization methods such as the stochastic gradient method \cite{stepSize,stepSize2} and the stochastic successive convex approximation method in \cite{THCF,RTHP,MJBP,TTS_OAC}. Second, Assumption 1-(3) states that the diminishing speed of $\xi^{t}$ is slower than that of $\omega^{t}$. According to \eqref{appRs}-\eqref{appRk}, we can see that the average SINR vector $\hat{\bm{\gamma}}_{A}$ is roughly obtained by averaging the instantaneous SINR $\sum_{j=1}^{J}\text{\ensuremath{\tilde{\bm{\gamma}}}}_{A}(\bm{\theta}^{t-\frac{1}{\xi^{t}}}$, $\bm{b}_{j}^{t-\frac{1}{\xi^{t}}-1})/J$, $\sum_{j=1}^{J}\text{\ensuremath{\tilde{\bm{\gamma}}}}_{A}(\bm{\theta}^{t-\frac{1}{\xi^{t}}+1}$, $\bm{b}_{j}^{t-\frac{1}{\xi^{t}}})/J,...,\sum_{j=1}^{J}\text{\ensuremath{\tilde{\bm{\gamma}}}}_{A}(\bm{\theta}^{t-2},\bm{b}_{j}^{t-1})/J$, $\sum_{j=1}^{J}\text{\ensuremath{\tilde{\bm{\gamma}}}}_{A}(\bm{\theta}^{t-1},\bm{b}_{j}^{t})/J $
over a time window of size $\frac{1}{\xi^{t}}$. Since $\bm{\theta}^{t}$ is changing over time $t$, the approximate average SINR vector $\tilde{\bm{\gamma}}_{A}$ may not converge to $\hat{\bm{\gamma}}_{A}$ in general. However, if $\lim_{t\rightarrow\infty}\omega^{t}/\xi^{t}=0$, it follows from \eqref{update_theta} that $\bm{\theta}^{t}$ is almost unchanged within the time window $\frac{1}{\xi^{t}}$ for sufficiently large t. In other words, $\tilde{\bm{\gamma}}_{A}$ will converge to $\hat{\bm{\gamma}}_{A}$ as $ t\rightarrow\infty$, which is crucial for guaranteeing the convergence of BSPD to a stationary point of problem (8). A typical choice $\xi^t$ and $\omega^t$ that satisfies
Assumption \ref{assumption1} is $\xi^=O(t^{-\kappa_1})$ and $\omega^=O(t^{-\kappa_2})$,  where $0.5<\kappa_1<\kappa_2\leq1$.


Based on Assumption \ref{assumption1}, we first prove a key lemma to establish the convergence of the recursive approximations $\tilde{\bm{\gamma}}^t_A$, $\bm{\Xi}^t$ and
$\bm{\nu}^t$ to true SINR vectors and partial derivations, which eventually leads to the final convergence result.
\begin{lemma}\label{lemma_2}
Under Assumption 1, we have
\begin{align}
\lim_{t\rightarrow\infty}|\tilde{\bm{\gamma}}^t_A-\hat{\bm{\gamma}}_A(\bm{\theta}^{t})|&=0,\label{lemma2_1}\\
\lim_{t\rightarrow\infty}|\tilde{f}^t(\bm{\theta}^t)-f(\hat{\bm{\gamma}}_A(\bm{\theta}^{t}))|&=0,\label{lemma2_2}\\
\lim_{t\rightarrow\infty} ||\bm{\Xi}^t-\nabla_{\bm{\theta}} f\big(\hat{\bm{\gamma}}_A(\bm{\theta}^{t})\big)||&=0,\label{lemma2_3}
\end{align}
\begin{align}
\lim_{t_1,t_2\rightarrow\infty} \tilde{f}^{t_1}&(\bm{\theta}^{t_1})-\tilde{f}^{t_2}(\bm{\theta}^{t_2})\leq \nonumber\\
&C\sqrt{\|\bm{\theta}^{t_1}-\bm{\theta}^{t_2}\|^2+\|\bm{\nu}^{t_1}-\bm{\nu}^{t_2}\|^2}.\label{lemma2_4}
\end{align}
where $C>0$ is a positive constant. Moreover, we consider a sequence $\{\bm{\theta}^{t_j}\}_{j=1}^{\infty}$ converging to a limiting point $\bm{\theta}^{\star}$, and define a
function
\begin{align}
\hat{f}(\bm{\theta})\triangleq
&\tilde{g}(\bm{\theta}^{\star})+\nabla_{\bm{\theta}}f(\hat{\bm{\gamma}}_A(\bm{\theta}^{\star}))(\bm{\theta}-\bm{\theta}^{\star})-q(\bm{\theta},\bm{\theta}^{\star}),
\end{align}
which satisfy $\hat{f}(\bm{\theta}^{\star})=f(\hat{\bm{\gamma}}_A(\bm{\theta}^{\star}))$ and
$\nabla_{\bm{\theta}}\hat{f}(\bm{\theta}^{\star})=\nabla_{\bm{\theta}}f(\hat{\bm{\gamma}}_A(\bm{\theta}^{\star})).$ Then, almost surely, we have
\begin{equation}
\lim_{t\rightarrow\infty} \tilde{f}^{t_j}(\bm{\theta})=\hat{f}(\bm{\theta}),\forall \bm{\theta}\in \bm{\Omega}.\label{lemma2_5}
\end{equation}
\end{lemma}

Please refer to Appendix \ref{appendixA} for the detailed proof. With Assumption 1 and Lemma \ref{lemma_2}, the following convergence theorem can be proved.

\begin{theorem}\label{convergTheorem}
Suppose Assumption \ref{assumption1} is satisfied. Let $\{\bm{\nu}^{t_j},\bm{\theta}^{t_j}\}$ denote any subsequence of iterates generated by Algorithm \ref{alg:BSPD} that
converges to a limiting point $(\bm{\nu}^{\star},\bm{\theta}^{\star})$. Then we almost surely have
\begin{align}
\bm{\nu}^{\star}=&\nabla_{\hat{{\bm{r}}}_A} f\left(\hat{{\bm{\gamma}}}^{\star}_A\right)\\
(\bm{\theta}-\bm{\theta}^{\star})^T\nabla_{\bm{\theta}} f\left(\hat{{\bm{\gamma}}}^{\star}_A\right)\leq& 0, \forall \bm{\theta}\in \mathbf{\Omega},\label{Theorem_2}
\end{align}
where $\hat{{\bm{\gamma}}}^{\star}_A=\hat{{\bm{\gamma}}}_A(\bm{\theta}^{\star})$.
\end{theorem}

Please refer to Appendix \ref{appendixB} for the detailed proof. Moreover, the limiting point $\bm{\theta}^{\star}$ generated by Algorithm \ref{alg:BSPD} also satisfies the
stationary
condition in \eqref{stationaryCon}. Therefore, we conclude that Algorithm \ref{alg:BSPD} converges to stationary solutions of problem $\mathcal{P}$.

Finally, we analyze the computational complexity of the proposed BSPD algorithm. The number of floating point operations (FPOs) is used to evaluate the complexity.  In each iteration of the proposed algorithm for solving problem \eqref{surroP}, we solve the subproblems for the two blocks of variables in a parallel and distributed manner.
\begin{enumerate}[(1)]
\item Let us focus on the subproblem with respect to $\bm{v}$. Notwithstanding the computation of the invariant terms, the complexity for evaluating the expression of $\bm{v}^{t,\star}$ in \eqref{opt_v} is dominated by three parts. The first part calculates the intermediate variable $\mathbf{\varUpsilon}$ with complexity $O(JK^{2}M)$. The second part applies the matrix inverse based on Gaussian Jordan elimination with complexity $O(M^{3})$. The third part utilizes the bisection method to search the Lagrangian parameter $\mu$ with complexity $\delta_{s}\triangleq\log_{2}(\frac{\varsigma_{0,s}}{\varsigma_{s}})$, where $\varsigma_{s}$ is the initial size of the search interval and $\varsigma_{0,s}$ is the tolerance. Thus the computational complexity for updating $\bm{v}$ is $O(M^{3}+JK^{2}M+\delta_{s})$.
\item  Let us focus on the subproblem with respect to $\bm{u}$. Following a similar approach, we can obtain the computational complexity for updating $\bm{u}$, i.e., $O(N^{3}+JKM)$.
\end{enumerate}

We usually have $M>N$.  Based on the above analysis, the computational complexity of solving problem \eqref{surroP} is $O(M^{3}+JK^{2}M+\delta_{s})$. Overall, the complexity of the proposed BSPD algorithm is given by $O(I_{1}M^{3}+I_{1}JK^{2}M+I_{1}\delta_{s})$, where $I_{1}$ is the number of iterations. In addition, we remark that each subproblem is solved in a parallel and distributed manner, which helps to further reduce the total computational time.

\section{Simulation Results}\label{sec_simulation}
In this section, we use Monte Carlo simulations to demonstrate the benefits of the proposed stochastic transceiver scheme in terms of the SINR utility. For all
simulations, unless otherwise specified, the following set of parameters are used. Consider a single-cell SR system with $K=4$ Tags, where the simulation topology is
illustrated in Fig. \ref{fig:tpology}. Specifically, the PT and the PR are respectively located at (100,0) and (100,200) meters, while the Tags are randomly distributed at
region [99,101]$\times$[198,199.5] meters.
\begin{figure}[!h]
\centering
\includegraphics[width=8cm]{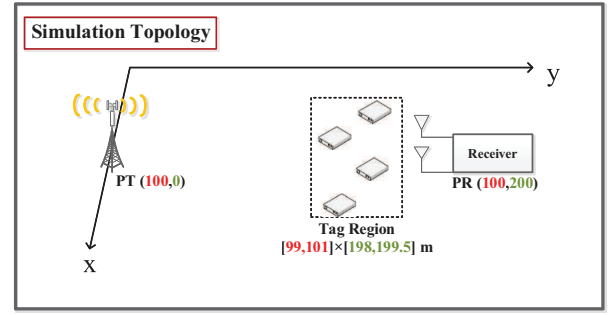}
\caption{An illustration of the simulation topology.}\label{fig:tpology}
\end{figure}

As in \cite{SymRadio,rate_region,ResourceA_BSC,FD_BSC,Cooperative_BSC,simulation_parameter}, the corresponding channel coefficients from the PT to the PR/each Tag  are
generated as normalized independent Rayleigh fading components with distance-dependent path loss modeled as $\mathrm{PL}(d_0)=\frac{\lambda^2_c G_T
G_R}{(4\pi)^2(d_0)^{\chi_0}}$ and  $\mathrm{PL}_k(d_k)=\frac{\lambda^2_c G_T G_{B,k}}{(4\pi)^2(d_k)^{\chi_k}}$, where $\lambda_c=0.33$ m is the carrier signal wavelength,
corresponding to $900$ MHz; $G_T$, $G_R$, and $G_{B,k}$ are the antenna gains for the PT, the PR, and Tag $k$, respectively; $d_0$ and $d_k$  is the distance between the PT
and the PR/Tag $k$ in meters, respectively; $\chi_0$ and $\chi_k$ are the corresponding path loss exponent. The channel from each Tag to the PR is assumed to be static, and
mainly depends on the large-scale path loss $\overline{\mathrm{PL}}_k(d_{b,k})=\frac{\lambda^2_c G_{B,k} G_R}{(4\pi)^2 (d_{b,k})^{\chi_{b,k}}}$, where $d_{b,k}$ is the
distance between Tag $k$ and the PR in meters, and $\chi_{b,k}$ is the path loss exponent for Tag $k$. We set $L=128$, $\chi_0=\chi_k=3.5$, $\chi_{b,k}=2$, and
$G_T=G_R=G_{B,k}=6$ dB \cite{Antenna_Gain}. The channel bandwidth is $1$ MHz. Furthermore, we consider $M=64$ antennas for the PT, and $N=16$ antennas for the PR. The noise
power at the PR is $\sigma^2_w=-100$ dBm. The transmit power budget for the PT is $P^{\mathrm{max}}=10$ dBm.
The size of each mini-batch is $J=10$. In addition, the reflection coefficients of all Tags are assumed to be same, i.e., $\alpha_k=0.1, \forall k$, and the activity probability of all
Tags' symbol are set $\rho_k=0.5, \forall k$.
All results are obtained by averaging over $200$ independent channel realizations. For the proposed BSPD algorithm, we set the
hyper-parameters as $\xi^t=\frac{1}{(1+t)^{\frac{2}{3}}}$ and $\omega^t=\frac{20}{20+t}$.
In our simulations, we use the barrier-function-based utility with  $\gamma_{0,k}=5$ as an example to illustrate the advantages of the proposed scheme.
Three alternative baseline schemes are included as follows:
\begin{itemize}
\item \textbf{Baseline 1}: In this scheme, the transmit beamforming is obtained as the eigenvector vectors of the DL channel from its eigenvalue decomposition (EVD).
    Furthermore, the receive beamforming is chosen to be the corresponding MMSE receiver.

\item \textbf{Baseline 2}: In this scheme, we first perform EVD operation to DL channel and all effective BL channels, i.e., $\mathbf{H}_d\mathbf{H}^H_d$ and
    $\hat{\bm{h}}_k\hat{\bm{h}}^H_k,\forall k$. Then, the transmit beamforming is obtained by the linear superposition of all resulting eigenvector vectors. The receive
    beamforming is chosen to be the corresponding MMSE receiver.

\item \textbf{Baseline 3} \cite{SymRadio}:  In this scheme, we schedule one Tag according to maximum channel gain criterion in each resource block.  Then, the transmit beamforming is obtained by the linear superposition of all resulting eigenvector vectors of the effective channel. The receive beamforming is chosen to be the corresponding MMSE receiver.

\end{itemize}

\begin{figure}[h]
\centering
\includegraphics[width=9cm]{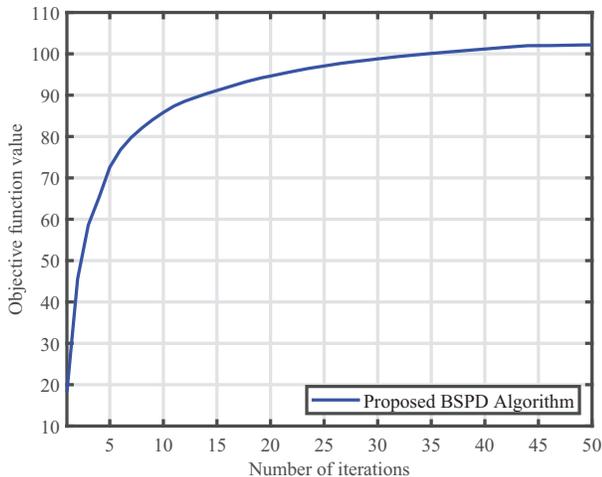}
\caption{Convergence behavior of the BPSD algorithm.}\label{fig:conv}
\end{figure}
In Fig. \ref{fig:conv}, we plot the objective function versus the iteration number, which illustrates the convergence behavior of the proposed BSPD algorithm. It is observed that the proposed BSPD algorithm quickly converges to stationary solution.

\begin{figure}[h]
\centering
\includegraphics[width=9cm]{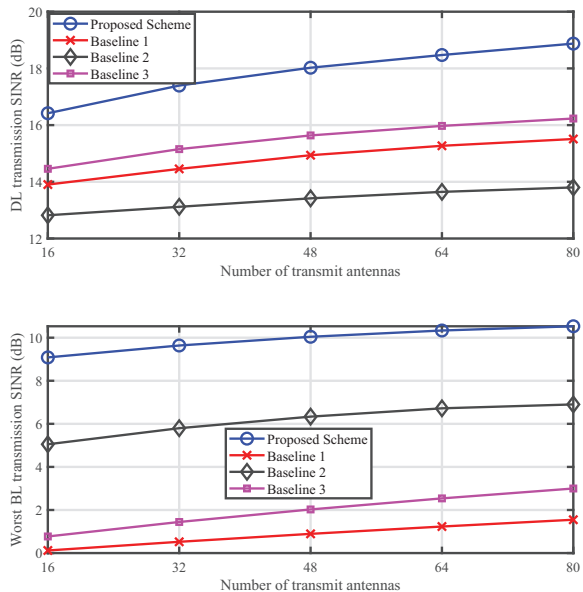}
\caption{Average SINR performance versus the number of transmit antennas $M$  when $N=16$ and $P^{\mathrm{max}}=10$ dBm.}\label{fig:varM}
\end{figure}
\begin{figure}[h]
\centering
\includegraphics[width=9cm]{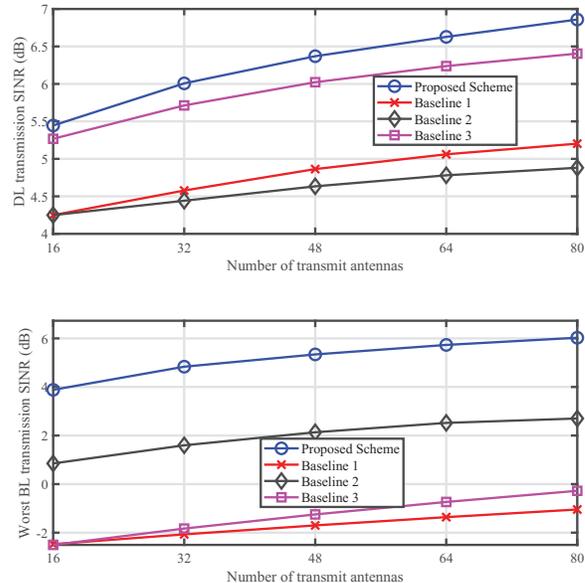}
\caption{Average SINR performance versus the number of transmit antennas $M$  when $N=16$ and $P^{\mathrm{max}}=0$ dBm.}\label{fig:varML}
\end{figure}

In Fig. \ref{fig:varM}, we plot the average SINR performance versus the number of transmit antennas $M$ when $N=16$ and $P^{\mathrm{max}}=10$ dBm.
The average SINR performance is evaluated in terms of both DL transmission SINR $\hat{\gamma}_s$ and worst BL transmission SINR $\hat{\gamma}_b=\min_{k} \hat{\gamma}_k$.
 It shows that the performance of all schemes is monotonically increasing with the number of transmit antennas; the growth rate tapers off as the number of transmit antennas increases. Moreover, it is seen that the proposed BSPD scheme achieves significant gain over all the other competing schemes for all $M$, which demonstrates the importance of the joint transceiver optimization.
The reason is that the proposed  BSPD scheme can exploit the difference in channel strengths among links to effectively assist both BL and DL transmission, while the baseline schemes do not take this difference into account. In addition, it is observed that as the number of transmit antennas increases, the gap between the proposed BSPD scheme and these other algorithms is enlarged. In Fig. \ref{fig:varML}, we plot the average SINR performance versus the number of transmit antennas $M$ when $N=16$ and $P^{\mathrm{max}}=0$ dBm. As expected, the proposed BSPD scheme achieve better average SINR performance over all the other competing schemes. This indicates that even in the low-SNR regime,  the judicious stochastic transceiver design still facilitate the symbiotic relationship between DL and BL transmissions.

\begin{figure}[h]
\centering
\includegraphics[width=9cm]{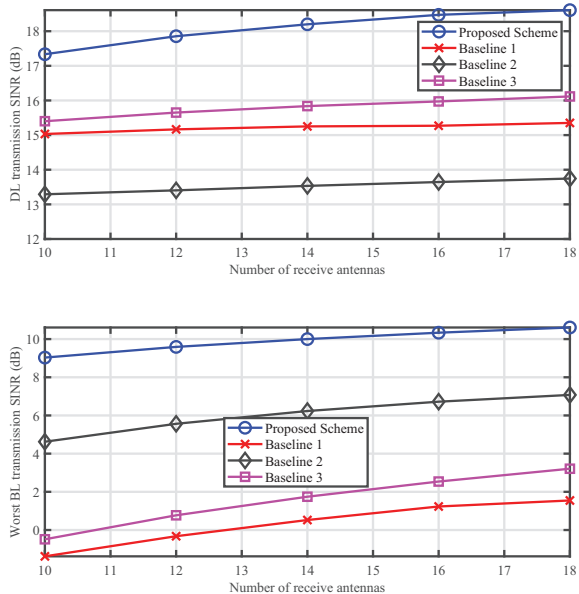}
\caption{Average SINR performance versus the number of receive antennas $N$  when $M=64$ and $P^{\mathrm{max}}=10$ dBm.}\label{fig:varN}
\end{figure}

In Fig. \ref{fig:varN}, we plot the average SINR performance versus the number of receive antennas $N$ when $M=64$ and $P^{\mathrm{max}}=10$ dBm. It is observed that the proposed BSPD scheme significantly outperforms all the other competing schemes, particularly for moderate and large number of receive antennas. Furthermore, we can see that the performance of the proposed BSPD scheme and Baseline 2 / 3 improves with the increase of $N$, due to the increased decoding capability. However, the performance of Baseline 1 is approximately invariant to the number of receive antennas. This is because Baseline 1 only focuses on improving the DL transmission, and can accurately decode the primary signal for a certain number of receive antennas. On the other hand, all Tags cannot exploit the primary system to realize opportunistic BL transmission.

\begin{figure}[h]
\centering
\includegraphics[width=9cm]{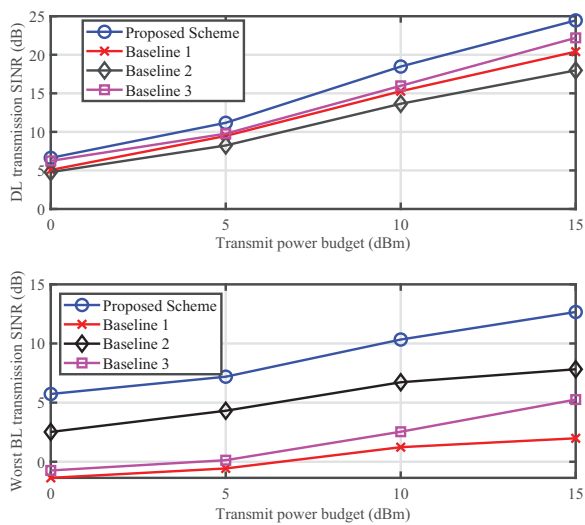}
\caption{Average SINR performance versus the transmit power budget (dBm)  when $N=16$ and $M=64$.}\label{fig:varSNR}
\end{figure}

In Fig. \ref{fig:varSNR}, we show the average SINR performance comparison versus the transmit power budget for  different schemes when $N=16$ and $M=64$. We can see that as the maximum transmit power increases, the average SINR of all schemes increases gradually. It is observed that the average SINR achieved by the proposed BSPD scheme is higher than that achieved by all the other baseline schemes, especially in the high transmit power budget regime. This indicates that  the proposed BSPD scheme can leverage the judicious transceiver design to
make full use of the entire system's radio resources, and further realizes better symbiotic relationship between DL and BL transmissions.

%


In Table \ref{bar}, we show the performance comparison of different schemes in terms of DL transmission rate and worst-case BL transmission BER when $M=64$, $N=16$, and $P^{\mathrm{max}}=10$ dBm. It is observed that the proposed BSPD scheme achieves better tradeoff performance among different links than other baselines.
Here, we remark that the worst-case BL transmission BER is based on the OOK coding scheme.
The reason is that the proposed BSPD scheme exploits the additional multipath components provided by all the Tags and further significantly enhance the throughput of the DL transmission. On the other hand, the  opportunistic BL transmission is enabled by the primary DL transmission. Therefore, there is an overall benefit achieved by the proposed BSPD scheme as compared to all the other competing schemes.

 \begin{table}[h]
  \centering
  \caption{Performance  comparison among different links for different schemes.}\label{bar}
 \begin{tabular}{|c|c|c|}
   \hline
   & DL transmission rate&Worst BL transmission BER\\
   \hline
   Proposed scheme&6.167 bps/Hz & 6.129$\times 10^{-3}$\\
   \hline
   Baseline 1 &$5.115$ bps/Hz& 1.3$\times 10^{-1}$\\
   \hline
    Baseline 2  & $4.594$ bps/Hz &0.5444$\times 10^{-1}$\\
   \hline
    Baseline 3 &$5.342$ bps/Hz &$1.175\times 10^{-1}$\\
   \hline
 \end{tabular}
 \end{table}

\section{Conclusion}\label{sec_con}
In this paper, we consider the stochastic transceiver design for the downlink transmission of the multi-Tags SR systems, to alleviate the performance bottleneck caused by the
DLI and inter-Tag interference. We formulate the optimization of transceiver design as the GNUMP under some practical constraints.
Based on the online observation of some random system states, we tailor a surrogate function with some specific structure and subsequently develop a novel algorithm named BSPD
to solve the resulting problem in a mini-batch fashion. In addition, we also show that the proposed BSPD algorithm converges to stationary solutions of the original GNUMP.
Finally, simulation results verify that the proposed algorithms can achieve significant gain over the baselines.
\appendices

\section{Proof of Lemma \ref{lemma_2}}\label{appendixA}
According to the law of large numbers and the central limit theorem, we can have
\begin{equation}
\tilde{\bm{\gamma}}^t_A \xrightarrow[]{\text{a.s.}}\hat{\bm{\gamma}}^t_A,~\mathbb{E}\|\tilde{\bm{\gamma}}^t_A-\hat{\bm{\gamma}}^t_A\|=O(\frac{1}{\sqrt{Jt}}),\label{largNum}
\end{equation}
where $\hat{\bm{\gamma}}^t_A=\hat{\bm{\gamma}}^t_A(\bm{\theta}^t)$. Consequently, \eqref{lemma2_1} and \eqref{lemma2_2} immediately follows from \eqref{largNum}.

Then, we focus on proving \eqref{lemma2_4} which characterizes the Lipschitz continuity of $\tilde{f}^t(\bm{\theta})$ with respect to $\hat{\bm{\gamma}}_A$ and $\bm{\theta}$.
Here, we recall that $\tilde{f}^t(\bm{\theta})$ is continuously differentiable functions of $(\hat{\bm{\gamma}}_A,\bm{\theta})$. Moreover, since channel samples are always bounded
in practice, the first-order and second-order derivative of $\tilde{f}^t(\bm{\theta})$ with respect to $(\hat{\bm{\gamma}}_A,\bm{\theta})$ are respectively bounded. In other words,
$\tilde{f}^t(\bm{\theta})$ is Lipschitz continuous with respect to $\hat{\bm{r}}_A$ and $\bm{\theta}$, respectively. Therefore, it follows that \eqref{lemma2_4} holds.

In addition, \eqref{lemma2_3} is a consequence of \cite{Feasible_SP}, Lemma 1, which provides a general convergence result for any sequences of random vectors
satisfies conditions (a)-(e) in this lemma. In the following, we aim to verify that the technical conditions (a)-(e) are satisfied therein. Since the instantaneous SINRs $\gamma_s$
and $\gamma_k$ are bound, we can find a convex and closed box region to contain $\gamma_s$ and $\gamma_k$ such that condition (a)-(b) are satisfied. From \eqref{largNum}, we have
\begin{align}
||\mathbb{E}[\bm{\Xi}^t]-&\nabla_{\bm{\theta}} f\big(\hat{\bm{\gamma}}_A(\bm{\theta}^{t})\big)||\leq \nonumber\\ &\mathbb{E}\|\mathbf{J}_{\gamma}(\bm{\theta}^t)(\nabla_{\hat{\bm{\gamma}}_A}
f\big(\tilde{\bm{\gamma}}_A(\bm{\theta}^{t}))-\nabla_{\hat{\bm{\gamma}}_A} f\big(\hat{\bm{\gamma}}_A(\bm{\theta}^{t}))\|\nonumber\\
&\overset{(a)}{=}O(\|\tilde{\bm{\gamma}}^t_A-\hat{\bm{\gamma}}^t_A\|)=O(\frac{1}{\sqrt{Jt}}),\label{conditionC}
\end{align}
where \eqref{conditionC}-(a) follows from the fact that the first-order derivative of $ f\big(\hat{\bm{\gamma}}_A(\bm{\theta}^{t}))$ is Lipschitz continuous and
$\mathbf{J}_{\gamma}(\bm{\theta}^t)$ are bounded w.p.1. Together \eqref{conditionC} with $\lim_{t\rightarrow\infty} \xi^t t^{-\frac{1}{2}}<\infty$, we have $
\lim_{t\rightarrow\infty} \xi^t ||\mathbb{E}[\bm{\Xi}^t]-\nabla_{\bm{\theta}} f\big(\hat{\bm{\gamma}}_A(\bm{\theta}^{t})\big)||<\infty.$ As such, the technical condition (c) in
\cite{Feasible_SP}, Lemma 1 is satisfied. Moreover, the technical condition (d) also follows from the assumption on $\{\xi^t\}$. Based on \eqref{lemma2_4} and $\lim_{t
\rightarrow \infty} \omega^t/\xi^t=0$, it is easy to verify that the technical condition (e) is satisfied.

Follows from \eqref{lemma2_1}-\eqref{lemma2_4}, \eqref{lemma2_5} immediately holds. This completes the proof.

\section{Proof of Theorem \ref{convergTheorem}}\label{appendixB}
For clarity, we let $\bm{\zeta}=[\bm{\nu}^T,\bm{\theta}^T]^T$ denote the composite control variables. When there is no ambiguity, we use $f(\bm{\phi})$ as an abbreviation for
$f(\hat{\bm{\gamma}}_A(\bm{\theta}))$.
\subsubsection{Step 1}
We first prove that $\lim \inf_{t\rightarrow\infty}\|\tilde{\bm{\zeta}}^t-\bm{\zeta}^t\|$=0 w.p.1.
From \eqref{surrogate}, $\tilde{f}^t(\bm{\theta})$ is uniformly strongly concave, and thus
\begin{equation}\label{AppendixB_Lip}
\nabla^T\tilde{f}^t(\bm{\theta}^t)\bm{d}^t\geq \eta\|\bm{d}^t\|^2+\tilde{f}^t(\tilde{\bm{\theta}}^t)-\tilde{f}^t(\bm{\theta}^t)\geq
\eta\|\bm{d}^t\|^2,
\end{equation}
where $\bm{d}^t=\tilde{\bm{\theta}}^t-\bm{\theta}^t$, and $\eta>0$ is some constant. From Assumption \ref{assumption1} and the fact that the gradient of ${f}^t(\bm{\zeta}^t)$
is Lipschitz continuous, there exists $L_f>0$ such that
\begin{align} \label{Appendix_low}
{f}(\bm{\zeta}^t)&\overset{(a)}{\geq}{f}(\bm{\zeta}^{t-1})+
\omega^t \nabla^T_{\bm{\theta}}{f}^t(\bm{\zeta}^{t-1})\bm{d}^{t-1}\nonumber\\
&-L_f(\omega^t)^2\|\bm{d}^{t-1}\|^2-O(\omega^t) \nonumber\\
&={f}^t(\bm{\zeta}^{t-1})-L_f(\omega^t)^2\|\bm{d}^{t-1}\|^2-O(\omega^t)\nonumber\\
&+\omega^t(\nabla^T_{\bm{\theta}}{f}^t(\bm{\zeta}^{t-1})-
\nabla^T \tilde{f}^t(\bm{\theta}^{t-1})+\nabla^T \tilde{f}(\bm{\theta}^{t-1}))\bm{d}^{t-1}\nonumber\\
&\overset{(b)}{\geq}{f}^t(\bm{\zeta}^{t-1})+\omega^t\eta\|\bm{d}^{t-1}\|^2-O(\omega^t),
\end{align}
where $O(\omega^t)$ implies that $\lim_{t\rightarrow\infty}O(\omega^t)/\omega^t=0$, \eqref{Appendix_low}-(a) follows from the first-order Taylor expansion of
$\tilde{f}^t(\bm{\zeta}^t)$, and \eqref{Appendix_low}-(b) follows from \eqref{AppendixB_Lip} and
$\lim_{t\rightarrow\infty}\|\nabla^T_{\bm{\theta}}{f}^t(\bm{\zeta}^{t-1})-
\nabla^T \tilde{f}^t(\bm{\theta}^{t-1})\|=0$.

In the following, we prove $\lim \inf_{t\rightarrow\infty}\|\tilde{\bm{\zeta}}^t-\bm{\zeta}^t\|=0$ w.p.1 by contradiction.
\begin{proof}
Suppose there exist a positive constant $\chi>0$ such that $\lim_{t\rightarrow\infty}\|\nabla^T_{\bm{\theta}}{f}^t(\bm{\zeta}^{t-1})-
\nabla^T \tilde{f}^t(\bm{\theta}^{t-1})\|\geq\chi$ with a certain probability. Then we can find a realization such that $\|\bm{d}^{t-1}\|\geq \chi$ for all $t$. W.l.o.g. we
focus on such a realization. By choosing a sufficiently large $t_0$, there exists $\overline{\varphi}>0$ such that
\begin{align}\label{contra_appendixB}
{f}(\bm{\zeta}^t)-{f}(\bm{\zeta}^{t-1})\geq \omega^t \overline{\varphi}\|\bm{d}^{t-1}\|^2,\forall t\geq t_0.
\end{align}
Moreover, we have
\begin{align}
{f}(\bm{\zeta}^t)-{f}(\bm{\zeta}^{t-1})\geq \overline{\varphi}\chi^2\sum_{j=t_0}^t(\omega^j)^2,
\end{align}
which, in view of $\sum_{j=t_0}^{\infty}(\omega^j)^2=\infty$, contradicts the boundedness of $\{{f}(\bm{\zeta}^t)\}$. As such, $\lim
\inf_{t\rightarrow\infty}\|\tilde{\bm{\zeta}}^t-\bm{\zeta}^t\|=0$ must be held.
\end{proof}
\subsubsection{Step 2}
Then we focus on proving that $\lim \sup_{t\rightarrow\infty}\|\tilde{\bm{\zeta}}^t-\bm{\zeta}^t\|=0$ w.p.1.
The proof relies on the following useful lemma.
\begin{lemma}\label{appendixB_lemma3}
There exists a constant $\hat{L}>0$ such that
\begin{align}
\|\tilde{\bm{\theta}}^{t_1}-\tilde{\bm{\theta}}^{t_2}\|\leq \hat{L}\|\bm{\theta}^{t_1}-\bm{\theta}^{t_2}\|+e(t_1,t_2),
\end{align}
where $\lim_{t_1,t_2\rightarrow\infty}=0$.
\end{lemma}

\begin{proof}
Following a similar proof of Lemma \ref{lemma_2}, it can be shown that
\begin{align}
\lim_{t\rightarrow\infty}|\tilde{f}^t(\bm{\theta})-\overline{f}^t(\bm{\theta};\bm{\zeta}^t)|&=O(\overline{e}_{\overline{t}}),\label{lemma3_1}\\
\lim_{t\rightarrow\infty}\|\tilde{\nu}^t-\nabla_{\hat{\bm{\gamma}}_A}f(\hat{\bm{\gamma}}_A^t)\|&=O(\overline{e}_{\overline{t}}),\label{lemma3_2}
\end{align}
where $\overline{f}^t(\bm{\theta};\bm{\zeta}^t)\triangleq
f(\overline{\bm{\gamma}}_A^t)+\nabla_{\bm{\theta}}f(\bm{\zeta}^t)(\bm{\theta}-\bm{\theta}^t)-\tau\|\bm{\theta}-\bm{\theta}^t\|^2
$ and $\lim_{\overline{t}}\overline{e}_{\overline{t}}\rightarrow0$. Note that $\overline{f}^t(\bm{\theta};\bm{\zeta}^t)$ and $\nabla_{\hat{\bm{\gamma}}_A}f(\hat{\bm{r}}_A^t)$ are
Lipschitz continuous in $\bm{\zeta}^t$, which can be easily verified. Furthermore, we have
\begin{align}
|\overline{f}^t(\bm{\theta};\bm{\zeta}^{t_1})-\overline{f}^t(\bm{\theta};\bm{\zeta}^{t_2})|&\leq D \|\bm{\zeta}^{t_1}-\bm{\zeta}^{t_2}\|,\label{lemma3_3}\\
\|\nabla_{\hat{\bm{\gamma}}_A}f(\hat{\bm{\gamma}}_A^{t_1})-\nabla_{\hat{\bm{\gamma}}_A}f(\hat{\bm{\gamma}}_A^{t_2})\|&\leq D \|\bm{\zeta}^{t_1}-\bm{\zeta}^{t_2}\|,\forall \bm{\theta}\in
\bm{\Omega},\label{lemma3_4}
\end{align}
for some positive constant $D>0$. Based on \eqref{lemma3_1}-\eqref{lemma3_4}, we have
\begin{align}
|\tilde{f}^{t_1}(\bm{\theta})-\tilde{f}^{t_2}(\bm{\theta})|&\leq D\|\bm{\theta}^{t_1}-\bm{\theta}^{t_2}\|+O(\overline{e}_{\overline{t}})+e(t_1,t_2),\label{lemma3_5}\\
\|\tilde{\bm{\nu}}^{t_1}-\tilde{\bm{\nu}}^{t_2}\|&\leq D \|\bm{\theta}^{t_1}-\bm{\theta}^{t_2}\|+O(\overline{e}_{\overline{t}})+e(t_1,t_2),\label{lemma3_6}
\end{align}
where $\lim_{t_1,t_2\rightarrow\infty}e(t_1,t_2)=0$. Since \eqref{lemma3_5} holds for any $\overline{t}>0$ and $\lim_{\overline{t}}\overline{e}_{\overline{t}}\rightarrow0$, we
have
\begin{align}
|\tilde{f}^{t_1}(\bm{\theta})-\tilde{f}^{t_2}(\bm{\theta})|&\leq D\|\bm{\theta}^{t_1}-\bm{\theta}^{t_2}\|\!+\!e(t_1,t_2),\forall \bm{\theta}\in \bm{\Omega}.\label{lemma3_7}
\end{align}
Combining the Lipschitz continuity and strong concavity of $\tilde{f}^t(\bm{\theta})$ with \eqref{lemma3_7}, it can be shown that
\begin{align}
\|\tilde{\bm{\theta}}^{t_1}-\tilde{\bm{\theta}}^{t_2}\|\leq D_1D_2\|\bm{\theta}^{t_1}-\bm{\theta}^{t_2}\|+D_1e(t_1,t_2),\label{lemma3_8}
\end{align}
for some positive constants $D_1,D_2>0$. This is because for strictly concave problem, when the objective function \eqref{surrogate} is changed by
amount $e(\bm{\theta})$, the optimal solution $\tilde{\bm{\theta}}^t$ will be changed by the same order (i.e.,  $\pm O(\mid e(\bm{\theta})\mid)$). Finally,
Lemma \ref{appendixB_lemma3} follows from \eqref{lemma3_4} and \eqref{lemma3_8}.
\end{proof}
According to Lemma \ref{appendixB_lemma3} and the same analysis as that in \cite{stepSize}, Proof of Theorem 1, it can be shown that
\begin{align}
\lim \sup_{t\rightarrow\infty}\|\tilde{\bm{\zeta}}^t-\bm{\zeta}^t\|=0,~\mathrm{w.p.1.} \label{lemma3_9}
\end{align}

\subsubsection{Step 3}
In the rest of the proof, we are ready to prove the convergence theorem. By definition, we have $\lim_{j\rightarrow\infty}
\tilde{\bm{\nu}}^{t_j}=\nabla_{\hat{\bm{\gamma}}_A}f(\hat{\bm{\gamma}}^{\star}_A).$
Furthermore, it follows from \eqref{lemma3_9} that $\bm{\nu}^{\star}=\lim_{j\rightarrow\infty} {\bm{\nu}}^{t_j}=\nabla_{\hat{\bm{r}}_A}f(\hat{\bm{r}}^{\star}_A).$ Combing
\eqref{surroP}, Lemma \ref{lemma_2} and \eqref{lemma3_9}, $\tilde{\bm{\theta}}^{\star}$ must be the optimal solution of the following concave optimization problem w.p.1.:
\begin{align}
\mathop{\max}_{\bm{\theta}\in \bm{\Omega}}~\tilde{f}^t(\bm{\theta}).
\end{align}
From the first-order optimality condition, we have
\begin{align}
\nabla^T \tilde{f}^t(\bm{\theta}^{\star})(\bm{\theta}-\bm{\theta}^{\star})\leq0,\forall \bm{\theta}\in \bm{\Omega}.\label{lemma3_10}
\end{align}
From Lemma \ref{lemma_2} and \eqref{lemma3_10}, $\bm{\theta}^{\star}$ also satisfies \eqref{Theorem_2}. This completes the proof.

\end{document}